\documentclass[a4paper]{article}
\pdfoutput=1
\usepackage{longtable}
\usepackage{multirow}

\usepackage{amsmath}
\usepackage{amsfonts}
\usepackage{amsthm}
\usepackage{amssymb}
\usepackage{cases}
\usepackage{bm}

\usepackage{color}
\usepackage[usenames,dvipsnames]{xcolor}

\makeatletter
\renewcommand*\env@matrix[1][*\c@MaxMatrixCols c]{%
  \hskip -\arraycolsep
  \let\@ifnextchar\new@ifnextchar
  \array{#1}}
\makeatother

\newcounter{mnotecount}[section]

\newcommand{\mnotex}[1]
{\protect{\stepcounter{mnotecount}}$^{\mbox{\footnotesize $\bullet$\themnotecount}}$ 
\marginpar{
\raggedright\tiny\em
$\!\!\!\!\!\!\,\bullet$\themnotecount: #1} }

\newtheorem{lem}{Lemma}
\newtheorem{thm}{Theorem}
\newtheorem{defn}[thm]{Definition}  
\newtheorem{prop}[thm]{Proposition} 

\usepackage{tikz}
\usepackage{tikzit}
\usetikzlibrary{shapes,arrows}
\usetikzlibrary{arrows.meta}

\usepackage[colorlinks=true,linkcolor=blue,citecolor=red,urlcolor=Violet,bookmarksdepth=subsection,final]{hyperref}

  \newcommand{\oo}{\infty}
  \newcommand{\del}{\partial}
  \newcommand{\sgn}{\operatorname{sgn}}
  
\renewcommand{\d}{\mathrm{d}}
  
  \newcommand{\Lie}{\mathcal{L}}

  \newcommand{\bell}{\bm{\ell}}
  \newcommand{\bn}{\mathbf{n}}
  \newcommand{\bem}{\mathbf{m}}
  \newcommand{\bA}{\mathbf{A}}
  \newcommand{\cA}{\mathcal{A}}
  \newcommand{\bB}{\mathbf{B}}
  \newcommand{\bC}{\mathbf{C}}
  \newcommand{\bD}{\mathbf{D}}
  \newcommand{\cF}{\mathcal{F}}
  \newcommand{\cG}{\mathcal{G}}
  \newcommand{\Ia}{I_{\mathrm{a}}}
  \newcommand{\Ib}{I_{\mathrm{b}}}
  \newcommand{\Ic}{I_{\mathrm{c}}}
  \newcommand{\Id}{I_{\mathrm{d}}}
  \newcommand{\Ie}{I_{\mathrm{e}}}
  \newcommand{\ff}{\mathrm{f}}
  \newcommand{\bK}{\mathbf{K}}
  \newcommand{\cK}{\mathcal{K}}
  
  \newcommand{\bL}{\mathbf{L}}
  \newcommand{\fL}{\mathfrak{L}}
  \newcommand{\bM}{\mathbf{M}}
  \newcommand{\bN}{\mathbf{N}}
  \newcommand{\bp}{\mathbf{p}}
  \newcommand{\bT}{\mathbf{T}}
  \newcommand{\bX}{\mathbf{X}}
\renewcommand{\Im}{\operatorname{Im}}
\renewcommand{\Re}{\operatorname{Re}}
  
\def\beq#1\eeq{\begin{align}#1\end{align}}

\title{IDEAL characterization of vacuum pp-waves}

\author{Igor Khavkine, David McNutt, Lode Wylleman}

\date{}

\begin{document}
\maketitle

\begin{abstract}

An IDEAL characterization of a particular spacetime metric, $g_0$, consists of a set of tensorial equations $T[g] = 0$ arising from expressions constructed from the metric, $g$, its curvature tensor and its covariant derivatives and which are satisfied if and only if $g$ is locally isometric to the original metric $g_0$. Earlier applications of the IDEAL classification of spacetimes relied on the construction of particular scalar polynomial curvature invariants as an important step in the procedure. 

In this paper we investigate the well-known class of vacuum pp-wave spacetimes, where all scalar polynomial curvature invariants vanish, and determine the applicability of an IDEAL classification for these spacetimes. We consider a modification of the IDEAL approach which permits a corresponding extension of the Stewart-Walker lemma. With this change, we are able to construct invariants and IDEAL-ly classify all of the vacuum pp-wave solutions which admit a two- or higher-dimensional isometry group, with the exception of one case. 

\end{abstract}

\section{Introduction}

An old topic in General Relativity~\cite[Ch.9]{kramer} \cite{mars-local}
is the question of how to decide when two explicit solutions of
Einstein's equations are actually the same solution up to a change of
coordinates, or conversely are distinct even allowing for a change of
coordinates. These questions are special cases of \emph{invariant
equivalence} and \emph{characterization} within the broader scope of
Riemannian, pseudo-Riemannian, and more general kinds of
geometries~\cite{olver-eis}. One traditional approach to such questions
is the so-called \emph{Cartan-Karlhede algorithm}~\cite{karlhede}
\cite[\S9.2]{kramer}, which assigns a normalized
frame, specially adapted to the algebraic structures of the Weyl
curvature tensor and Ricci tensor, to the spacetime. An alternative, which has gained some attention
recently~\cite{fs-schw, fs-kerr, fs-typeD, fs-sphsym, fs-szsa, gpgl-schw,
gpgl-kerr, gpgl-typed, krongos-torre, cdk, kh-gst}, is the so-called
\emph{IDEAL}%
	\footnote{The acronym, explained in~\cite{fs-sphsym} (footnote, p.2),
	stands for Intrinsic, Deductive, Explicit and ALgorithmic.} %
\emph{characterization} approach, based on algebraic identities between
(possibly tensorial) curvature invariants of a spacetime in question,
which does not require auxiliary geometric structures other than the
metric itself. An advantage of the Cartan-Karlhede approach is its
generality and breadth of applicability. On the other hand, IDEAL
characterizations have the advantage of being intimately connected to
linearized gauge-invariant observables~\cite{kh-calabi,fhk}, intrinsic
characterization of a spacetime by its initial data~\cite{gpgl-kerr}, as
well as quantitative estimates of convergence to a target
spacetime~\cite{bobbgst, okounkova}.

Naively, the easiest way to compare two metrics $(g_1, M_1)$ and $(g_2,
M_2)$ is through their \emph{scalar curvature invariants}. In fact,
picking a certain number of scalars $\mathcal{R}[g] =
(\mathcal{R}^1[g],\ldots, \mathcal{R}^N[g])$ constructed by total
contractions of copies of the metric $g$, its curvature tensor
$R_{abcd}[g]$, and its covariant $\nabla_e$-derivatives, one can
immediately conclude that the metrics on domains $U_1 \subset M_1$ and
$U_2 \subset M_2$ are not isometric if the images of the maps
$\mathcal{R}[g_1]\colon U_1 \to \mathbb{R}^N$ and
$\mathcal{R}[g_2]\colon U_2 \to \mathbb{R}^N$ do not coincide. For
Riemannian metrics, under mild uniformity assumptions, one can
conversely conclude that the equality of these images implies the
existence of an isometry between $(g_1|_{U_1}, U_1)$ and $(g_2|_{U_2},
U_2)$, for special choices of $N$ and $\mathcal{R}[g]$~\cite{chp-4d,
chp-higherd}.

In the Lorentzian case, there exist large families of spacetime metrics
that cannot be distinguished by any scalar curvature invariants, while
containing non-isometric members~\cite{chp-4d, chp-higherd}. It is for
this reason that the Cartan-Karlhede algorithms supplements scalar
cuvature invariants with invariants with respect to a normalized frame.
Analogously, the IDEAL approach makes additional use of tensorial
curvature invariants. The algorithmic aspects of the Cartan-Karlhede
approach have been extensively studied and its scope of applicability
and limitations are well-understood~\cite[\S9.2]{kramer}. No systematic
approach has yet been developed for IDEAL characterizations and they
have been so far found only by ad-hoc methods for special families of
spacetime geometries.

\emph{Vacuum pp-waves}~\cite{EhlersKundt} \cite[\S24.5]{kramer} are a
family of non-flat spacetimes with special algebraic and symmetry
properties. All possible scalar curvature invariants of a vacuum pp-wave
vanish, yet the family contains mutually non-isometric subfamilies that are
parametrized by free continuous parameters and even functions. These spacetimes belong to the larger class of VSI spacetimes for which all scalar polynomial curvature invariants vanish \cite{pravda2002all}. 

This property along with the pp-waves interpretation as a toy-model for gravitational waves has lead to the study of these solutions in string theory and alternative gravity theories \cite{horowitz1990strings}. The vacuum pp-waves provide a simple illustration of the significant difference in the approaches to the local characterization of Riemannian metrics and Lorentzian metrics using invariants constructed from the metric, the curvature tensor and its covariant derivatives.

Vacuum pp-waves have been extensively studied in the Cartan-Karlhede
approach, including a complete classification~\cite{mcnutt-thesis,
mcnutt-ppwave1, mcnutt-ppwave2}. To date, no IDEAL characterization of
vacuum pp-waves has been carried out. In this work, we partially rectify
this omission and give an as complete as possible IDEAL characterization
of isometry classes within the $4$-dimensional vacuum pp-wave subfamily
with more than the minimum number of Killing vectors (\emph{highly
symmetric} pp-waves). See Theorem~\ref{thm:ideal-pp-wave-isom} for the
precise statement, including the description of the characterized
subfamily and the small list of exceptions. We continue to use ad-hoc
methods in this case, made difficult by the impossibility to use any
non-trivial scalar curvature invariants. We hope that the experience
gained in this work will help in formulating a systematic approach to
IDEAL characterizations.

In Section~\ref{sec:ideal}, we recall the precise notion of an IDEAL
characterization and introduce some ways that we were forced to generalize
it. Section~\ref{sec:pp-wave} contains relevant background geometric
information about $4$-dimensional vacuum pp-wave spacetimes. In
Section~\ref{sec:main} we present our main results including the
classification of isometry classes within the highly symmetric subfamily
(Theorem~\ref{thm:pp-wave-isom-classes},
Table~\ref{tab:pp-wave-isom-classes}, Figure~\ref{fig:pp-wave-classes}),
and the IDEAL characterization of thse isometry classes
(Theorem~\ref{thm:ideal-pp-wave-isom}, Table~\ref{tab:pp-wave-ideal},
Figure~\ref{fig:pp-wave-flowchart}. Section~\ref{sec:discuss} finishes
with a discussion and an outlook toward future work.

\section{IDEAL characterization} \label{sec:ideal}

As was alluded to in the Introduction, an IDEAL characterization of a
spacetime is a local invariant condition that uses only the metric tensor and
quantities that can be covariantly obtained from it. More specifically,
an ideal characterization of a geometry $(M,g_0)$ is a list of tensorial
equations
\begin{equation} \label{eq:ideal-example}
	\bT_k[g] = 0, \quad k=1,2,\ldots,N,
\end{equation}
constructed covariantly out of a metric $g$ and its derivatives
(concomitants of the Riemann tensor) that are satisfied on a small
neigborhood of a point if and only if that neighborhood is isomorphic to
a neighborhood of a point of $(M,g_0)$. In this case, we say that the
tensor identities~\eqref{eq:ideal-example} characterize the local
isometry class of $(M,g_0)$. We also allow the possibility that $g_0$ is
a family of (possibly non-isometric) metrics and the tensor
identities~\eqref{eq:ideal-example} imply that $g$ is locally isometric
to some member of that family. In such cases, we speak of an IDEAL
characterization of the whole family.

The most classical example is of spaces of locally constant curvature,
characterized by the identity
\begin{equation} \label{eq:ideal-cc}
	R_{abcd} = \alpha \left(g_{ac} g_{bd} - g_{ad} g_{bc}\right) ,
\end{equation}
where $\alpha$ is a constant; in relativity it is re-expressed via the
cosmological constant $\Lambda = \frac{(n-1)(n-2)}{2} \alpha$, in dimensions $n>2$.
Of course $\alpha=0$ corresponds to the locally flat case. The
\emph{locally} attribute is important because~\eqref{eq:ideal-cc} as any
IDEAL characterization~\eqref{eq:ideal-example} is a fundamentally local
condition, which is insensitive to the global topology of the spacetime.
Other notable geometries that have been given IDEAL characterizations
include Schwarzschild~\cite{fs-schw} (including in higher
dimensions~\cite{kh-gst}) and Kerr~\cite{fs-kerr} black holes, as well
as FLRW~\cite{cdk} cosmological solutions. A growing list of other
examples can be gleaned from the references in the Introduction.

For the sake of terminology, let \emph{IDEAL tensors} refer to the
left-hand sides of~\eqref{eq:ideal-example}. If $v$ is a vector field
and $\bT[g]$ is tensor covariantly constructed from the metric, then the
identity (sometimes known as the \emph{Stewart-Walker
lemma}~\cite{sw-pert}) $\dot{\bT}_g[\Lie_v g] = \Lie_v \bT[g]$, where
$\bT[g+h] = \bT[g] + \dot{\bT}_g[h] + O(h^2)$, immediately links an
IDEAL characterization to linear gauge invariant observables. Namely, if
$\bT[g]$ is one of the IDEAL tensors for $g$, then $\dot{\bT}_g[\Lie_v
g] = 0$ implies that $\dot{\bT}_g[h]$ is invariant under local
linearized gauge transformations (diffeomorphisms) of linearized gravity
about the background $g$. In fact, linearizing all the IDEAL tensors for
a given $g$ yields a good candidate for a complete set of linear gauge
invariants on that background~\cite{kh-calabi, fhk, kh-compat}. We
consider this connection to linear gauge invariants through
linearization to be one of the important applications of IDEAL
characterizations.

Later, we will find that the limitations on how IDEAL tensors~$\bT[g]$
are to be constructed to be slightly restrictive. For instance, $\bA[g]$
and $\bB[g]$ may be two such tensors, of equal valence say, that become
proportional, $\bA[g] = X[g] \bB[g]$, for some subclass of metrics. Taking
derivatives and tensor products of $\bA[g]$ and $\bB[g]$, it might be
possible to impose some additional limited set of polynomial conditions
on the scalar $X[g]$ and its derivatives, but non-polynomial conditions
would certainly be impossible. That is a significant limitation while
trying to characterize all isometry classes of pp-wave spacetimes.
However, supposing that the proportionality factor $X[g]$ exists, and
using it in further non-polynomial conditions seems to be sufficient to
do the job. A natural question then: is this step in the
spirit of IDEAL characterizations?

The E in the acronym stands for Explicit, so with $X[g]$ being only
implicitly defined, the strict answer seems to be No. However, if $X[g]$
exists, it is defined uniquely. Moreover, as we shall prove below,
when $X[g]$ is defined, it still satisfies the Stewart-Walker lemma, thus
preserving the connection with linear gauge invariants that we described
above. Let us make the last statement more precise with the following
lemma, where a scalar $X[g]$ is just a special case of a tensor field
$\bX[g]$.

\begin{lem}[extended Stewart-Walker] \label{lem:st-ext}
Let $\bA$ and $\bB$ be tensors covariantly constructed
from the metric. Suppose the background metric $g_{ab}$ is chosen so
that $\bB\ne 0$ and the relation $\bA = \bX \otimes
\bB$ holds for some tensor $\bX$. Then there exists a
differential operator $\dot{\bX}[h]$ acting on metric
perturbations $h_{ab}$ such that $\Lie_v \bX =
\dot{\bX}[\Lie_v g]$ for any vector field $v^a$.
\end{lem}

\begin{proof}
Let us denote the linearization by a dot, as in $\bA_{g+h} =
\bA_g + \dot{\bA}_g[h] + O(h^2)$. In general,
$\bA$, $\bB$ and $\bX$ may all have different
numbers of indices (including no indices), say $\bA_{c_1\cdots
c_k : a_1\cdots a_j}$ and $\bB_{b_1\cdots b_j}$, with
$\bX_{c_1\cdots c_k}$, and where $:$ is used only to visually
separate groups of indices. Let us define the operation
\begin{equation}\label{def-wedge-AB}
	(\bA \wedge \bB)_{c_1\cdots c_k : a_1\cdots a_j : b_1\cdots b_j}
		= \bA_{c_1\cdots c_k : a_1\cdots a_j} \bB_{b_1\cdots b_j}
		- \bA_{c_1\cdots c_k : b_1\cdots b_j} \bB_{a_1\cdots a_j} .
\end{equation}
By elementary linear algebra, the relation $\bA = \bX
\otimes \bB$ between $\bA$ and $\bB$ holds iff
$\bA \wedge \bB$ = 0. It is also convenient to note the
antisymmetry identity $(\bX\otimes \bB) \wedge \mathbf{C}
= -(\bX\otimes \mathbf{C}) \wedge \bB$. Linearizing the
identity $\bA \wedge \bB = 0$ about the chosen background
metric $g_{ab}$, we find (dropping all $g$ subscripts)
\begin{equation}
	\dot{\bA} \wedge \bB + \bA \wedge \dot{\bB}
	= 0 .
\end{equation}
Using antisymmetry and bilinearity of $\wedge$, as well as the background
relation $\bA = \bX \otimes \bB$, we arrive at
\begin{equation}
	(\dot{\bA} - \bX \otimes \dot{\bB}) \wedge \bB = 0 ,
\end{equation}
which implies that there must exist an operator $\dot{\bX}[h]$
such that $\dot{\bA} - \bX \otimes \dot{\bB} =
\dot{\bX} \otimes \bB$. The notation $\dot{\bX}$ is
just formal, since $\bX$ itself is not defined for arbitrary
background metrics.

Next, we fix a vector field $v^a$ and recall the Stewart-Walker lemma,
$\Lie_v\bA = \dot{\bA}[\Lie_v g]$ and $\Lie_v\bB =
\dot{\bB}[\Lie_v g]$. Taking the Lie derivative of the relation
between $\bA$ and $\bB$ obtains the relation
\begin{multline}
	0 = \Lie_v (\bA - \bX \otimes \bB)
	\\
	= (\dot{\bA}[\Lie_v g] - \bX \otimes \dot{\bB}[\Lie_v g])
		- (\Lie_v \bX) \otimes \bB
	= (\dot{\bX}[\Lie_v g] - \Lie_v \bX) \otimes \bB ,
\end{multline}
which implies that $\Lie_v \bX = \dot{\bX}[\Lie_v g]$,
since $\bB \ne 0$.
\end{proof}

When $\bX[g]$ is defined implicitly as in Lemma~\ref{lem:st-ext}, we
will nevertheless refer to the corresponding $\dot{\bX}_g[h]$ as the
\emph{linearization} of $\bX[g]$ about the metric $g_{ab}$. Thus, since
the connection to linear gauge invariants through linearization is
important, for the purposes of this work a characterization of a spacetime geometry will still be considered to be in the spirit of
IDEAL characterization if it is IDEAL {\em except} for explicitly using
an implicitly defined tensor like $\bX[g]$. 

In light of the extended Stewart-Walker lemma, we will refer to a tensors defined implicitly from two or more IDEAL tensors as

\begin{defn}[conditional invariants]
Given two tensors ${\bf A}$ and ${\bf B}$ which are covariantly constructed from the metric, if there is a tensor of lower rank, ${\bf X}$, such that ${\bf A} = {\bf X} \otimes {\bf B}$, then ${\bf X}$ is a \emph{conditional invariant (tensor)}. If the rank is zero, then $X$ is a \emph{conditional invariant (scalar)}. 
\end{defn}

\noindent Conditional invariants, tensors ${\bf X}$ or scalars $X$, will be determined from tensors constructed from the curvature tensor and its covariant derivatives. 

\section{Vacuum pp-waves} \label{sec:pp-wave}

Plane-fronted gravitational waves with parallel rays,
in short \emph{pp-waves}, can be defined as follows.
\begin{defn} \label{def:pp-wave}
A {\em pp-wave} is a spacetime that admits a covariantly
constant null vector field, $\ell_a$:
\begin{equation}
	\nabla_a \ell_b = 0, \quad \ell^a \ell_a = 0 .
\end{equation}
\end{defn}
These spacetimes made their first appearance in the work of
Brinkmann~\cite{brinkmann-pp}. Subsequently, they have been widely
studied~\cite[\S 24.5]{kramer}. For the rest of this section, we
restrict to $4$ spacetime dimensions. In $4$~dimensions, an Einstein
vacuum pp-wave has the line element~\cite[\S 24.5]{kramer}
\begin{equation} \label{eq:pp-wave-ds2}
	\d{s}^2 = 2\d\zeta\d\bar{\zeta} - 2\d{u}\d{v} - 2H \d{u}^2 ,
	\quad
	H \equiv H(u,\zeta,\bar{\zeta}) = f(\zeta,u) + \bar{f}(\bar{\zeta},u) ,
\end{equation}
where $u,v$ are two real coordinates and $\zeta = x + i y$ is a
convenient complex coordinate, while $H(u,\zeta,\bar{\zeta})$ is real
harmonic in $\zeta$ and jointly smooth, justifying its presentation in
terms of the complex holomorphic in $\zeta$ and jointly smooth
$f(\zeta,u)$, $f_{,\bar{\zeta}} = 0$. The $1$-forms
\begin{equation}
	\bell = \d{u} , \quad
	\bn = \d{v} + H \d{u} , \quad
	\bem = \d\zeta , \quad
	\bar{\bem} = \d\bar{\zeta}
\end{equation}
constitute a complex null coframe dual to the compex null frame of
vector fields

\begin{equation}
\ell = -\del_v, \quad
n = - \del_u + H \del_v, \quad
m = \del_{\bar{\zeta}}, \quad
\bar{m} = \del_{\zeta}, 
\end{equation}

\noindent normalized as $\ell^a n_a = n^a \ell_a = -1$, $m^a \bar{m}_a = \bar{m}^a
m_a = 1$, and all other contractions vanishing.

After fixing the form~\eqref{eq:pp-wave-ds2} of the metric, the residual
gauge freedom (coming from coordinate transformations or
underdeterminedness of the parameters) consists of the continuous
transformations
\begin{subequations} \label{eq:pp-wave-residual-gauge}
\begin{gather}
	\notag
	\zeta' = e^{i\beta}(\zeta + h(u)) , \quad
	u' = (u+u_0)/a , \quad
	v' = a\, [v + \del_u (h(u) \bar{\zeta} + \bar{h}(u) \zeta) + g(u)] , \\
	f'(\zeta',u') = a^2 \, \left[f(\zeta,u) - \del_u^2\bar{h}(u) \zeta
		+ \tfrac{1}{2}\left|\del_u \bar{h}(u)\right|^2 - \tfrac{1}{2}\del_u
		g(u) + ic(u)\right] ,
	\label{eq:pp-wave-ds2-redef}
\end{gather}
where $\beta$, $a$, $u_0$ are real constants, $c(u)$ and $g(u)$ are
smooth real functions, and $h(u)$ is a smooth complex function, and the
discrete transformation
\begin{equation} \label{eq:pp-wave-ds2-redef-discrete}
	\zeta' = \bar{\zeta} , \quad
	v' = v , \quad
	u' = u , \quad
	f' = \bar{f} ,
\end{equation}
\end{subequations}
where $\bar{f}(\zeta,u) = \overline{f(\bar{\zeta},u)}$. The free choice
of $h(u)$ or $c(u)$ and $g(u)$, respectively, allows us to adjust
both $\del_u^2\bar{h}(u)$ and $\frac{1}{2} \left|\del_u
\bar{h}(u)\right|^2 - \frac{1}{2} \del_u g(u) + ic(u)$ to be
arbitrary smooth complex functions of $u$.

The above line element~\eqref{eq:pp-wave-ds2} is consistent with
$\ell^a$ being covariantly constant, hence in particular a Killing
vector. Thus, the spacetime admits at least a $G_1$ group of isometries
\cite{kramer}, which is the maximal isometry group for generic $f$.

\subsection{Curvature invariants} \label{sec:curv-inv}

Denote by $R_{ab}$ the Ricci tensor, by $C_{abcd}$ the Weyl tensor, and
by $\eta_{abcd}$ the Levi-Civita $4$-form. In addition, let us define
the \emph{complex self-dual Weyl tensor}
\begin{equation} \label{eq:sd-weyl}
	\bC^\dagger = \bC - i {}^*\bC , \quad \text{where} \quad
	{}^* C_{abcd} = \frac{1}{2} \eta_{ab ef} C^{ef}{}_{cd} ,
\end{equation}
and the \emph{Bel-Robinson tensor} \cite{Bel58,Bel62}
\begin{equation}
	T_{abcd}
	= C_{aecf} C_{b}{}^{e}{}_{d}{}^{f}
		+ {^*}C_{aecf} {^*}C_{b}{}^{e}{}_{d}{}^{f}
	= C^\dagger_{aecf} \bar{C}^\dagger{}_{b}{}^{e}{}_{d}{}^{f} ,
\end{equation}
which is fully symmetric $T_{abcd} = T_{(abcd)}$. The spacetime is
called \emph{complex recurrent} when the complex self-dual Weyl
tensor satisfies
\begin{equation}
	\nabla \bC^\dagger = \bK \otimes \bC^\dagger
\end{equation}
for some (unique) $1$-form $\bK$. A null vector field $\ell^a$ is itself
called \emph{recurrent} when
\begin{equation} \label{def:l-recur}
	\nabla \bell = \bp \otimes \bell
	\quad \text{or equivalently} \quad
	\ell_{[a} \ell_{b];c} = 0 ,
\end{equation}
for some (unique) $1$-form $\bp$.

The complex self-dual Weyl tensor of vacuum pp-waves takes the form
\begin{equation} \label{eq:bC-expr}
	\bC^\dagger
		= 8 f_{,\zeta\zeta} (\d{u}\wedge\d\zeta) \otimes (\d{u}\wedge\d\zeta)
		= 8 f_{,\zeta\zeta} (\bell\wedge\bem) \otimes (\bell\wedge\bem) .
\end{equation}
The special case $f_{,\zeta\zeta} = 0$ corresponds to Minkowski space. We exclude this case from further consideration, meaning that we assume from
now on assume that $f$ is not a linear polynomial in $\zeta$. Hence $f_{,\zeta\zeta}\neq 0$ almost everywhere, and henceforth we implicitly consider the spacetime points where this inequality holds. At those points the Weyl tensor is of Petrov type~\textbf{N}, $\bell$ spanning the unique quadruple \emph{principal null direction (PND)} \cite[$\S$ 4.3]{kramer}.
The Bel-Robinson tensor takes the form~\cite{bonilla1997}
\begin{equation} \label{eq:T-decomp-null}
	T_{abcd} = \beta \ell_a \ell_b \ell_c \ell_d , \quad
	\beta = 4 \left|f_{,\zeta\zeta}\right|^2 .
\end{equation}
Direct
calculation shows that $\bC^\dagger$ is complex recurrent, with
recurrence $1$-form
\begin{equation} \label{eq:bK-expr}
	\bK = \d(\ln f_{,\zeta\zeta})
	= \frac{(f_{,\zeta\zeta u} \bell + f_{,\zeta\zeta\zeta} \bem)}{f_{,\zeta\zeta}} ,
\end{equation}
which clearly satisfies $K^a \ell_a = 0$. For reference, it is useful to
compute
\begin{multline}
	\nabla {\bK}
		= \left((\ln f_{,\zeta\zeta})_{,uu}
				- \bar{f}_{,\bar{\zeta}} (\ln f_{,\zeta\zeta})_{,\zeta}\right)
					\bell \otimes \bell \\
			+ (\ln f_{,\zeta\zeta})_{,\zeta u} (\bell \otimes \bem + \bem \otimes \bell)
			+ (\ln f_{,\zeta\zeta})_{,\zeta\zeta} \bem \otimes \bem ,
\end{multline}
\begin{multline}
	\nabla\nabla {\bK}
		= \left((\ln f_{,\zeta\zeta})_{,uuu}
				- 3\bar{f}_{,\bar{\zeta}} (\ln f_{,\zeta\zeta})_{,\zeta u}\right) \bell \bell \bell
			+ (\ln f_{,\zeta\zeta})_{,\zeta uu} (\bem \bell \bell
					+ \bell \bell \bem + \bell \bem \bell)
	\\
			+ (\ln f_{,\zeta\zeta})_{,\zeta\zeta u} (\bem \bell \bem
					+ \bem \bem \bell + \bell \bem \bem)
			+ (\ln f_{,\zeta\zeta\zeta})_{,\zeta\zeta} \bem \bem \bem 
	\\
			- \bar{f}_{,\bar{\zeta}} (\ln f_{,\zeta\zeta})_{,\zeta\zeta}
				(\bem\bell\bell + \bell\bell\bem + \bell\bem\bell)
			- (\ln f_{,\zeta\zeta})_{,\zeta}
				(\bar{f}_{,\bar{\zeta} u} \bell\bell\bell
					+ \bar{f}_{,\bar{\zeta}\bar{\zeta}} \overline{\bem}\bell\bell)
\end{multline}

The vector $\mathbf{K}$ is a conditional invariant tensor and is the main geometric quantity that we will use
to build further invariants that will be used in later sections. We introduce these derived quantities here.  To do this, it will be necessary at times to contract the vector $\bK$ with various tensors and this will be denoted as
\begin{equation}
	\bK \cdot {\bf A} = A_{a_1...a_k} K^{a_k} \label{eq:dotproduct} .
\end{equation}
In general the conditional invariants will be
rational functions, so they will be well-defined only when their
denominators are non-vanishing.

The simplest invariant is just
\begin{equation} \label{eq:K2-def}
	|\bK|^2 := \overline{\bK} \cdot {\bK}
		= \left|(\ln f_{,\zeta\zeta})_{,\zeta}\right|^2 .
\end{equation}
From \eqref{eq:T-decomp-null}, \eqref{eq:bK-expr} and \eqref{def-wedge-AB} we have that
\begin{equation}\label{K-cond-G5G6}
|\bK|^2=0 \quad \Leftrightarrow \quad 
f_{,\zeta\zeta\zeta}=0
\quad \Leftrightarrow\quad \bT \wedge \bK = 0, 
\end{equation}
which implies $\bT \wedge (\nabla \bK)=0$.
In this case we define the following conditional invariants:
\begin{align}
\label{eq:Ia2-def}
	\left(\Ia^{(2)}\right)^4 \bT &:= 16 \bK^{\otimes 4} , &
	\Ia^{(2)} &= \frac{2^{\frac{1}{2}}}{|f_{,\zeta\zeta}|^{\frac{1}{2}}}
		(\ln f_{,\zeta\zeta})_{,u} , \\
\label{eq:Ib2-def}
	\left(\Ib^{(2)}\right)^2 \bT &:= 16 (\nabla\bK)^{\otimes 2} , &
	\Ib^{(2)} &= \frac{2}{|f_{,\zeta\zeta}|}
		(\ln f_{,\zeta\zeta})_{,uu} .
\end{align}
The meaning of the superscripts will become clear later. Strictly
speaking, these formulas only define the invariant scalars $\Ia^{(2)}$
and $\Ib^{(2)}$ as raised to the powers $\left(\Ia^{(2)}\right)^4$ and
$\left(\Ib^{(2)}\right)^2$. We must choose a branch of the square and
quartic roots to get rid of those powers. In principle, a branch choice
could be made separately at each point of $M$. However, once a branch
choice is made at a point $x\in M$, requiring continuity is sufficient
to ensures a unique branch choice on a sufficiently small contractible
neighborhood of $x$, provided of course that the invariant does not
vanish at $x$ and on this neighborhood. For later convenience in the
proof of Theorem~\ref{thm:ideal-pp-wave-isom}, we choose the square root
branch at $x$ so that $\Im \Ib^{(2)} \ge 0$, with a cut on the positive
real axis and $\Ib^{(2)} < 0$ on the cut (namely $(+1)^{1/2} = -1$); we
choose the quartic root branch so that $\Im \Ia^{(2)} \ge 0$ and $\Re
\Ia^{(2)} \le 0$, with a cut on the positive real axis and $\Ia^{(2)} <
0$ on the cut (namely $(+1)^{1/4} = -1$).

The next definitions for conditional invariant tensors and scalars will require $|\bK| \ne 0$:
\begin{align}
\label{eq:D-def}
	D_{ac}
	&:= \frac{\overline{K}^b \overline{K}^d C^\dagger_{abcd}}
			{|\bK|^4}
	= 2 \frac{f_{,\zeta\zeta}^3}{f_{,\zeta\zeta\zeta}^2} \ell_a \ell_c ,
	\\
\label{eq:J-def}
	J \bT &:= |\bK|^4 (\mathbf{D})^{\otimes 2} , \quad
	J = \frac{f_{,\zeta\zeta}^3}{\bar{f}_{,\bar{\zeta}\bar{\zeta}}^3}
		\frac{\bar{f}_{,\bar{\zeta}\bar{\zeta}\bar{\zeta}}^2}{f_{,\zeta\zeta\zeta}^2} ,
	\\
\label{eq:Ic-def}
	\Ic &:= \frac{\overline{\mathbf{K}}\cdot (\overline{\bf K}\cdot \nabla {\bf K})}{|\bK|^4}
	= -\left(\frac{1}{(\ln f_{,\zeta\zeta})_{,\zeta}}\right)_{,\zeta} ,
	\\
\label{eq:fL-def}
	\fL &:= \frac{\overline{\bK}\cdot \nabla\bK}{|\bK|^2} - \Ic \bK
		= \left(\frac{(\ln f_{,\zeta\zeta})_{,u}}{(\ln f_{,\zeta\zeta})_{,\zeta}}\right)_{,\zeta} \bell ,
	\\
\label{eq:Ia0-def}
	\left(\Ia^{(0)}\right)^4 \bT &:= 4 |\bK|^4 \fL^{\otimes 4} , \quad
		\Ia^{(0)} = \left|\frac{f_{,\zeta\zeta\zeta}^2}{f_{,\zeta\zeta}^3}\right|^{\frac{1}{4}}
				\left(\frac{(\ln f_{,\zeta\zeta})_{,u}}{(\ln f_{,\zeta\zeta})_{,\zeta}}\right)_{,\zeta} ,
	\\
\label{eq:Ib0-def}
	\left(\Ib^{(0)}\right)^2 \bT &:= |\bK|^4 \left(\nabla\fL
		- \frac{\overline{\bK}\cdot\nabla\fL}{|\bK|^2} \bK\right)^{\otimes 2} ,
	\notag \\
		\Ib^{(0)} &= \frac{1}{2} \left|\frac{f_{,\zeta\zeta\zeta}^2}{f_{,\zeta\zeta}^3}\right|^{\frac{1}{2}} \left[
				\left(\frac{(\ln f_{,\zeta\zeta})_{,u}}{(\ln f_{,\zeta\zeta})_{,\zeta}}\right)_{,\zeta u}
				- \left(\frac{(\ln f_{,\zeta\zeta})_{,u}}{(\ln f_{,\zeta\zeta})_{,\zeta}}\right)
				\left(\frac{(\ln f_{,\zeta\zeta})_{,u}}{(\ln f_{,\zeta\zeta})_{,\zeta}}\right)_{,\zeta\zeta}
			\right] ,
	\\
\label{eq:N-def}
	\overline{\bN}
		&= \nabla \frac{\overline{\bK}\cdot\nabla \bK}
				{|\bK|^2}
			-\left(\frac{\overline{\bK}\cdot\nabla \bK}
				{|\bK|^2}\right)^{\otimes 2}
			+\frac{|\bK|^2}{2} \overline{\bD} .
\end{align}
The branch choices for $\Ia^{(0)}$ and $\Ib^{(0)}$ are the same as for
$\Ia^{(2)}$ in~\eqref{eq:Ia2-def} and $\Ib^{(2)}$ in~\eqref{eq:Ib2-def}.
When $\bN \wedge \bD = 0$, we can also define the scalar
\begin{equation} \label{eq:scalN-def}
	N \bD := \bN .
\end{equation}
We will need the coordinate expression for $N$ only in some special pp-wave subcases. We note that $\fL$ is proportional to the covariantly constant null vector $\bell$ and so $\fL$ will be a recurrent null vector field. 

The next batch of definitions will require $|\bK|\ne 0$ and also $k(\Ic
- \tfrac{1}{2}) - \Ic \ne 0$ for some fixed constant $k\in \mathbb{C}$.
\begin{align}
\label{eq:Lk-def}
	\bL^{(k)} &:= \frac{\fL}{k(\Ic - \tfrac{1}{2}) - \Ic}
		= -2\Xi^{(k)} \bell ,
	\\
\label{eq:Xik-def}
	\Xi^{(k)} &:= \frac{1}{2} \frac{\left(\frac{(\ln f_{,\zeta\zeta})_{,u}}{(\ln f_{,\zeta\zeta})_{,\zeta}}\right)_{,\zeta}}
			{k\left[\left(\frac{1}{(\ln f_{,\zeta\zeta})_{,\zeta}}\right)_{,\zeta} + \frac{1}{2}\right]
				- \left(\frac{1}{(\ln f_{,\zeta\zeta})_{,\zeta}}\right)_{,\zeta}} ,
	\\
\label{eq:Idk-def}
	\Id^{(k)} \bD &:= -\frac{1}{2}\left(\bL^{(k)}\right)^{\otimes 2} , \quad
		\Id^{(k)} = -\left(\Xi^{(k)}\right)^2
			\frac{f_{,\zeta\zeta\zeta}^2}{f_{,\zeta\zeta}^3} ,
	\\
\label{eq:Iek-def}
	-\frac{1}{2}\overline{\Ie^{(k)}} |\bK|^2 \overline{\bD}
		&:= \nabla\bK - (\bK\fL + \fL\bK) - \Ic\bK\bK
	\notag\\ &\quad {}
			- (k-1) \left[k(\Ic - 1) - (\Ic+\tfrac{1}{2})\right] \bL^{(k)}\bL^{(k)} ,
	\\
\label{eq:Mk-def}
	\overline{\bM^{(k)}} &:= \left(\nabla - \frac{\bK}{|\bK|^2} \overline{\bK}\cdot\nabla
			- \frac{\overline{\bK}}{|\bK|^2} \bK\cdot\nabla\right)
				(\overline{\Ie^{(k)}} |\bK|^2 \overline{\bD})
		\notag \\ &\quad {}
			- \bL^{(k)} \left[(k-1)\Ic + 1\right] \overline{\Ie^{(k)}} |\bK|^2 \overline{\bD}
		\notag \\ &\quad {}
			- \frac{(k+1)}{2} \bL^{(k)} \left[-2|\bK|^2 \overline{\bD}
					+ k(k-1) \bL^{(k)}\bL^{(k)}\right] .
\end{align}
For any pp-wave, the right hand sides of the last three formulas will be
proportional to a product of $\bell$'s. All of the above formulas are
well-defined when $k=0$. On the other hand, most of them diverge or go
to zero as $k\to\oo$. It will be useful for us to define certain
rescaled limits of these formulas when $\Ic\neq \frac12$:
\begin{align}
\label{eq:Loo-def}
	\bL^{(\oo)} &:= \lim_{k\to\oo} k\bL^{(k)}
		= \frac{\fL}{\Ic - \frac{1}{2}} = -2\Xi^{(\oo)} \bell ,
	\\
\label{eq:Xioo-def}
	\Xi^{(\oo)} &:= \lim_{k\to\oo} k\Xi^{(k)}
		= \frac{1}{2} \frac{\left(\frac{(\ln f_{,\zeta\zeta})_{,u}}{(\ln f_{,\zeta\zeta})_{,\zeta}}\right)_{,\zeta}}
			{\left(\frac{1}{(\ln f_{,\zeta\zeta})_{,\zeta}}\right)_{,\zeta} + \frac{1}{2}} ,
	\\
\label{eq:Idoo-def}
	\Id^{(\oo)} \bD &:= -\frac12\left(\bL^{(\oo)}\right)^{\otimes 2} , \quad
		\Id^{(\oo)} = \lim_{k\to\oo} k^2 \Id^{(k)}
			= -\left(\Xi^{(\oo)}\right)^2 \frac{f_{,\zeta\zeta\zeta}^2}{f_{,\zeta\zeta}^3} ,
	\\
\label{eq:Ieoo-def}
	-\frac{1}{2}\overline{\Ie^{(\oo)}} |\bK|^2 \overline{\bD}
		&:= \lim_{k\to\oo} -\frac{1}{2} \overline{\Ie^{(k)}} |\bK|^2 \overline{\bD}
	\notag\\
		&= \nabla\bK - (\bK\fL + \fL\bK) - \Ic\bK\bK
			- (\Ic - 1) \bL^{(\oo)}\bL^{(\oo)} ,
	\\
\label{eq:Moo-def}
	\overline{\bM^{(\oo)}}
		&:= \lim_{\Re k=0, k\to\oo}\overline{\bM^{(k)}}
		\notag \\
		&:= \left(\nabla - \frac{\bK}{|\bK|^2} \overline{\bK}\cdot\nabla
			- \frac{\overline{\bK}}{|\bK|^2} \bK\cdot\nabla\right)
				(\overline{\Ie^{(\oo)}} |\bK|^2 \overline{\bD})
		\notag \\ &\quad {}
			- \Ic \bL^{(\oo)} \overline{\Ie^{(\oo)}} |\bK|^2 \overline{\bD}
			- \frac{1}{2} \bL^{(\oo)} \left[-2|\bK|^2 \overline{\bD}
					+ \bL^{(\oo)}\bL^{(\oo)}\right] .
\end{align}
In a similar way to $\bL$, the definitions of $\bL^{(k)}$ and $\bL^{(\infty)}$ are proportional to the covariantly constant null vector field $\bell$ and so they are both recurrent null vector fields, when they are defined. 
\subsection{Specialization to subclasses}

The exhaustive classification of the pp-wave classes listed in
Proposition~\ref{prop:pp-wave-classes} was first carried out in the classic work~\cite[\S 2--5.6]{EhlersKundt}, where Ehlers \& Kundt gave
a complete classification of vacuum pp-waves by their Lie algebras of
\emph{infinitesimal local isometries}%
	\footnote{At each point $x\in M$, the \emph{infinitesimal local
	isometries} of $(M,g)$ are the Killing vector fields that can be
	defined on at least some neighborhood of $x$. They obviously have some
	common domain of definition on which they have the structure of a Lie
	algebra, via the Lie bracket of vector fields}. %
For each class, they gave a representative
functional form of $f(\zeta,u)$ depending on constant or functional
parameters, which we summarize in the following
\begin{prop}[pp-wave classes] \label{prop:pp-wave-classes}
Let $(M,g)$ be a non-flat 4-dimensional vacuum pp-wave
(Definition~\ref{def:pp-wave}), where $M$ is connected. Then, assuming the Lie algebra of infinitesimal local isometries does not change at any point of $M$, the line element $g$ can be locally put into the
form~\eqref{eq:pp-wave-ds2}, with the function $f(\zeta,u)$ belonging to
one of the following canonical forms
\begin{equation} \label{eq:f-cases}
	f(\zeta,u) = \begin{cases}
		4\alpha u^{2i\kappa-2} \zeta^2 & G_{6a} \\
		e^{2i\lambda u} \zeta^2 & G_{6b} \\
		A(u) \zeta^2 & G_{5} 
			\\
		4\alpha e^{i\gamma} u^{-2} \ln\zeta & G_{3a} \\
		e^{i\gamma} \ln\zeta & G_{3b} \\
		e^{2\lambda\zeta} & G_{3c} \\
		e^{i\gamma} \zeta^{2i\kappa} & G_{3d} \\
		u^{-2} \ff(\zeta u^{i\kappa}) & G_{2a}
			\\
		\ff(\zeta e^{i\lambda u}) & G_{2b}
			\\
		A(u) \ln\zeta & G_{2c}
			\\
		f(\zeta,u) & G_{1}
	\end{cases} ,
\end{equation}
where $\alpha$, $\gamma$, $\lambda$ and $\kappa$ are real constants, $A(u)$ is an
arbitrary smooth complex function and $f(\zeta)$ is an arbitrary complex
holomorphic function. Moreover, these constant and functional parameters
are the same in each adapted chart on $M$, up to form-preserving
transformations allowed by~\eqref{eq:pp-wave-ds2-redef}.
\end{prop}

We note that in the $G_{3a}$ and $G_{3b}$ cases, the functional form differs by a phase factor which was not found in the original analysis by Ehler's and Kundt but was pointed out in Lemmas 8.9.1 and 8.9.4 of \cite{mcnutt-thesis}. Our results will cover the $G_{\text{2--6}}$ subclasses, excluding the
generic component of the $G_1$ class. We will refer to these geometries
as
\begin{defn}[highly-symmetric pp-waves] \label{def:high-sym}
Any pp-wave geometry covered by the classification of
Proposition~\ref{prop:pp-wave-classes} that has at least two independent
Killing vectors at each point, or equivalently belonging to one of the
$G_{\text{2--6}}$ subclasses, will be called \emph{highly symmetric}.
\end{defn}

\noindent Furthermore, we will impose an additional condition on the geometries studied here, in order to restrict the form of the functional form on the entire manifold. 

\begin{defn}[regular pp-wave] \label{def:regular}
A vacuum pp-wave geometry
$(M,g)$ is
\emph{regular} at $x\in M$ when the Lie algebra of infinitesimal local
isometries at each point of an open neighborhood of $x$ is isomorphic to
the one at $x$, and it satisfies the following conditions, broken
down by the symmetry classes listed in~\eqref{eq:f-cases}. Inequalities
are interpreted as holding at $x$ and equalities on at least some open
neighborhood of $x$ (note that $\wedge$ and $\vee$ respectively stand
for logical \emph{and} and \emph{or}):
\begin{center}
\begin{tabular}{rl}
	$G_{5}$: & $(B=\ln A)$ \\
		& $A \ne 0 \wedge \dot{B} \ne 0
		\wedge
		\left(\ddot{B} \ne 0 \vee \ddot{B} = 0\right)$ \\
		& ${}\wedge
		\biggl(\del_u(\Re\dot{B}e^{-\frac{1}{2}\Re B}) \ne 0
			\vee \Bigl[\del_u(\Re\dot{B}e^{-\frac{1}{2}\Re B}) = 0$ \\
		& $\qquad{}\wedge
			\left(\del_u(\Im\dot{B}e^{-\frac{1}{2}\Re B}) \ne 0
				\vee \del_u(\Im\dot{B}e^{-\frac{1}{2}\Re B}) = 0\right)\Bigr]\biggr)$ \\
	$G_{2a}$: & $\ff''(z) \ne 0 \wedge
		\left(\del_z\left( \frac{\ff'''(z)^2}{\ff''(z)^3} \right) = 0 \vee
			\del_z\left( \frac{\ff'''(z)^2}{\ff''(z)^3} \right) \ne 0\right)$
	\\
	$G_{2b}$: & $\ff''(z) \ne 0$ \\
		& ${}\wedge \biggl(
			\left[\lambda\ne 0 \wedge
			\Bigl(\del_z\left(\frac{\ff'''(z)^2}{\ff''(z)^3}\right) \ne 0
				\vee \del_z\left(\frac{\ff'''(z)^2}{\ff''(z)^3}\right) = 0\Bigr)\right]$
			\\
		& $\qquad{}\vee
			\left[\lambda = 0 \wedge
			\Bigl(\del_z\left(\frac{\ff'''(z)\ff'(z)}{\ff''(z)^2}\right) \ne 0
				\vee \del_z\left(\frac{\ff'''(z)\ff'(z)}{\ff''(z)^2}\right) = 0\Bigr)\right]
			\biggr)$
	\\
	$G_{2c}$: & $(B=\ln A)$ \\
		& $A \ne 0 \wedge \dot{B} \ne 0
		\wedge
		\left(\ddot{B} \ne 0 \vee \ddot{B} = 0\right)$ \\
		& ${}\wedge
		\biggl(\del_u(\Re\dot{B}e^{-\frac{1}{2}\Re B}) \ne 0
			\vee \Bigl[\del_u(\Re\dot{B}e^{-\frac{1}{2}\Re B}) = 0$ \\
		& $\qquad{}\wedge
			\left(\Im\dot{B} \ne 0 \vee \Im\dot{B} = 0\right)\Bigr]\biggr)$
\end{tabular}
\end{center}
If not explicitly listed above, the corresponding class is not
restricted. A \emph{regular} vacuum pp-wave geometry $(M,g)$ is one that
is highly-symmetric (Definition~\ref{def:high-sym}) and \emph{regular}
at each point $x\in M$.
\end{defn}

The logic behind the inequalities in Definition~\ref{def:regular} is as
follows. Some of them prevent the point $x\in M$ from falling onto a
branch point of the implicitly defined invariants $\Ia^{(0,2)}$ and
$\Ib^{(0,2)}$, while others simplify the use of the implicit function
theorem to find a relation between certain invariants in the proof of
Theorem~\ref{thm:pp-wave-isom-classes}.

In addition to the invariant quantities listed above, we will require some specialized invariants that appear only appear in particular subclasses. In the following we will consider the $G_5$ subclass along with the $G_{2a}$, $G_{2b}$ and $G_{2c}$ subclasses and their respective generalizations arising from including a $\zeta$-linear term and a $\zeta$-independent term to the original function. 

In the $G_5$ subclass, $f(\zeta,u) = A(u) \zeta^2$, let $A(u) =
e^{B(u)}$. Then we evaluate
\begin{equation} \label{eq:Iab2-eval}
	(\Ia^{(2)})^4 = \dot{B}^4 e^{-2\Re B} , \quad
	(\Ib^{(2)})^2 = \ddot{B}^2 e^{-2\Re B} .
\end{equation}
Similarly, in the $G_{2c}$ subclass, $f(\zeta,u) = A(u) \ln\zeta$, let
$A(u) = e^{B(u)}$. Then we evaluate
\begin{equation} \label{eq:Iab0-eval}
	(\Ia^{(0)})^4 = \dot{B}^4 e^{-2\Re B} , \quad
	(\Ib^{(0)})^2 = \ddot{B}^2 e^{-2\Re B} , \quad
	\frac{(\Ia^{(0)})^2}{\Id^{(0)}} = e^{i\Im B} .
\end{equation}
In the slightly more general class $f(\zeta,u) = A(u) \ln\zeta +
A_1(u) \zeta + A_0(u)$, we evaluate
\begin{equation} \label{eq:N-eval-G2c}
	|\bK|^2 = \frac{4}{|\zeta|^2} , \quad
	\bD = -\frac{A}{2} \bell\bell, \quad
	\bN = \frac{A_1}{\bar{\zeta}} \bell\bell, \quad \text{and} \quad
	N = -\frac{2}{\bar{\zeta}} \frac{A_1}{A} .
\end{equation}

In the $G_{2a}$ class, $f(\zeta,u) = u^{-2} \ff(\zeta u^{i\kappa})$.
Letting $z = \zeta u^{i\kappa}$, we then evaluate
\begin{equation} \label{eq:Ide-eval-G2a}
	|\bK|^2 = \left|\frac{\ff'''(z)}{\ff''(z)}\right|^2 , \quad
	\Id^{(i\kappa)} = -\frac{\ff'''(z)^2}{\ff''(z)^3} , \quad
	\Ie^{(i\kappa)} = (\ff'(z) + i\kappa(i\kappa-1) \bar{z})
		\frac{\ff'''(z)}{\ff''(z)^2} .
\end{equation}

We also evaluate the following invariant quantities in the generalization of $G_{2a}$, $f(\zeta,u) = u^{-2} \ff(\zeta
u^{i\kappa}) + A_1(u) \zeta + A_0(u)$:
\begin{equation} \label{eq:M-eval-G2a}
	\bL^{(i\kappa)} = -\frac{2}{u} \bell , \quad
	\bM^{(i\kappa)} = 2 \frac{\overline{\ff'''(z)}}{\overline{\ff''(z)}}
		u^{-i\kappa-1} (u\dot{A}_1 - (i\kappa-2) A_1) \bell\bell\bell .
\end{equation}

In the $G_{2b}$ class, $f(\zeta,u) = \ff(\zeta e^{i\lambda u})$.
Letting $z = \zeta e^{i\lambda u}$, we then have the following coordinate expressions:
\begin{equation} \label{eq:Ide-eval-G2b}
	|\bK|^2 = \left|\frac{\ff'''(z)}{\ff''(z)}\right|^2 , \quad
	\Id^{(\oo)} = \lambda^2 \frac{\ff'''(z)^2}{\ff''(z)^3} , \quad
	\Ie^{(\oo)} = (\ff'(z) - \lambda^2 \bar{z})
		\frac{\ff'''(z)}{\ff''(z)^2} .
\end{equation}
By considering the slight generalization, $f(\zeta,u) = \ff(\zeta
e^{i\lambda u}) + A_1(u) \zeta + A_0(u)$, we find the following
\begin{equation} \label{eq:M-eval-G2b}
	\bL^{(\oo)} = -2i\lambda \bell , \quad
	\bM^{(\oo)} = 2 \frac{\overline{\ff'''(z)}}{\overline{\ff''(z)}}
		e^{-i\lambda u} (\dot{A}_1 - i\lambda A_1) \bell\bell\bell .
\end{equation}
It will be helpful to consider a more specific case of this new function, $f(\zeta,u) = \frac{4\alpha}{u^2}
e^{i\gamma} \ln \zeta + A_1(u)\zeta + A_0(u)$ in order to compute a particular conditional invariant
\begin{equation} \label{eq:Id-eval-G3a0}
	\frac{1}{\Id^{(k)}} = \alpha e^{i\gamma} .
\end{equation}
We note that the dependence on $k$ disappears from this expression. 

In addition to the above quantities for their respective subclasses or generalizations, it will be useful to evaluate the $J$ invariant for several
subclasses. Namely, for $f(\zeta,u) = e^{i\gamma} \ln \zeta + A_1(u)\zeta +
A_0(u)$ or $f(\zeta,u) = \frac{4\alpha}{u^2} e^{i\gamma} \ln \zeta +
A_1(u)\zeta + A_0(u)$,
\begin{equation} \label{eq:J-eval-G3ab}
	J = e^{2i\gamma} ,
\end{equation}
for $f(\zeta,u) = e^{i\gamma} \zeta^{2i\kappa}$,
\begin{equation} \label{eq:J-eval-G3d}
	|\bK| = \left|\frac{2(i\kappa-1)}{\zeta}\right| , \quad
	J = \frac{(2i\kappa-1) (i\kappa+1)^2}{(2i\kappa+1) (i\kappa-1)^2}
		|\zeta|^{4i\kappa} e^{2i\gamma} ,
\end{equation}
for $f(\zeta,u) = \ff(\zeta)$,
\begin{equation} \label{eq:J-eval-G2b}
	J = \frac{\ff''(\zeta)^3}{\ff'''(\zeta)^2}
		\frac{\overline{\ff'''(\zeta)^2}}{\overline{\ff''(\zeta)^3}} .
\end{equation}

\subsection{Notation summary} \label{sec:notation}

For convenience, we summarize the invariants that we have introduced
above by their differential order and the non-vanishing conditions
needed for their definition
\begin{center}
\begin{tabular}{r|c|c|c|c}
	& $\del^2 f(\zeta,u)$ & $\del^3 f(\zeta,u)$ & $\del^4 f(\zeta,u)$ & $\del^5 f(\zeta,u)$
		\\
	\hline\hline
		& $\bC^\dagger$, $\bT$
		& $\bK$, $|\bK|$
		& $\Ia^{(2)}$, $\Ib^{(2)}$
		&
		\\
	$|\bK|\ne 0$ & 
		& $\mathbf{D}$, $J$
		& $\Ic$, $\fL$, $\Ia^{(0)}$, $\Ib^{(0)}$
		& $N$, $\bN$
		\\
	$\textstyle \Ic-\frac{1}{2} \ne 0$, $|\bK|\ne 0$ &
		&
		& $\bL^{(\oo)}$, $\Id^{(\oo)}$, $\Ie^{(\oo)}$
		& $\bM^{(\oo)}$
		\\
	$\textstyle k(\Ic-\frac{1}{2})\ne \Ic$, $|\bK|\ne 0$ &
		&
		& $\bL^{(k)}$, $\Id^{(k)}$, $\Ie^{(k)}$
		& $\bM^{(k)}$
\end{tabular}
\end{center}

For the $G_5$ and $G_{2c}$ subclasses:
\begin{center}
\begin{tabular}{r|c|c|c}
	& $B(u)$ & $\dot{B}(u)$ & $\ddot{B}(u)$
		\\
	\hline\hline
		&
		& $\Ia^{(2)}$
		& $\Ib^{(2)}$
		\\
	$|\bK|\ne 0$
		& $J$
		& $\Ia^{(0)}$
		& $\Ib^{(0)}$
\end{tabular}
\end{center}

For the $G_{2a}$ and $G_{2b}$ subclasses:
\begin{center}
\begin{tabular}{r|c|c}
	& $\ff''(z)$ & $\ff'''(z)$
		\\
	\hline\hline
		& & $|\bK|$
		\\
	$\textstyle \Ic-\frac{1}{2} \ne 0$, $|\bK|\ne 0$
		& $|\bK|^2/|\Id^{(\oo)}|$, $(\Ie^{(\oo)})^2/\Id^{(\oo)}$
		& $\Id^{(\oo)}$, $\Ie^{(\oo)}$
		\\
	$\textstyle k(\Ic-\frac{1}{2})\ne \Ic$, $|\bK|\ne 0$
		& $|\bK|^2/|\Id^{(k)}|$, $(\Ie^{(k)})^2/\Id^{(k)}$
		& $\Id^{(k)}$, $\Ie^{(k)}$
\end{tabular}
\end{center}

\section{Main results} \label{sec:main}

In 4 dimensions, vacuum pp-waves are traditionally classified by the
dimension and isomorphism type of the Lie algebra of infinitesimal local
isometries, ranging through $G_1$, $G_2$, $G_3$, $G_5$ and $G_6$, with
various subclasses~\cite[Tbl.24.2]{kramer}. The next increase in the
dimension of the symmetry algebra already gives $G_{10}$, corresponding
to the maximally symmetric locally flat case, which we exclude from
consideration. 

The various symmetry classes can be further refined into subclasses by a careful analysis of associated differential equations with their formfactors. The result of this is a complete list of formfactors which are uniquely characterized by invariant constant parameters and signature functions, whenever an arbitrary function is involved in the formfactor. This analysis is summarized in Figure~\ref{fig:pp-wave-classes} and all of the possible formfactors are listed in Table~\ref{tab:pp-wave-isom-classes}.

To connect the various formfactors of the pp-waves to an IDEAL classification, we begin with the most general case and determine a chain of conditions on tensor quantities. The binary decision of whether the condition is satisfied or not satisfied will either lead to a new tensor condition or an end node for which the sequence of conditions leading to the node fully characterize a particular subclass of the highly-symmetric pp-waves. This is summarized in Figure~\ref{fig:pp-wave-flowchart}.

\subsection{Four dimensional vacuum pp-waves} \label{sec:4d-vac-pp}

The vacuum pp-waves in four dimensions (4D) are characterized by the following
\begin{prop} \label{prop:ideal-pp-wave}
 A 4-dimensional spacetime is a vacuum pp-wave
 (Definition~\ref{def:pp-wave}) if and only if the following conditions
 hold at every point: 
\begin{enumerate}
	\item[(i)] $R_{ab}=0$,
	\item[(ii)] $C^\dagger_{ab}{}^{cd}C^\dagger_{cdef}=0$,
	\item[(iii)] $T_{abc[d}T_{e]fgh;i}=0$. 
\end{enumerate}	
\end{prop}

\begin{proof}
Condition (i) is the vacuum condition, while condition (ii) expresses that the
Weyl tensor is of Petrov type~N (see,
e.g.,~\cite{ferrando2001covariant,wylleman2021poynting} or \cite[\S 4.1]{kramer}). Let $\ell^a$ span the
unique PND. Then the Bel-Robinson tensor is of the
form~\cite{Bel58,Bel62} 
\begin{equation}\label{Bel-Rob-typeN}
	T_{abcd}=\beta \ell_a\ell_b\ell_c\ell_d, \quad \beta> 0.
\end{equation}

Given this form we easily obtain
\[
	T_{abc[d}T_{e]fgh;i}
	= \beta \ell_a\ell_b\ell_c\ell_f\ell_g\ell_h (\ell_{[d}\ell_{e];i}) .
\]
Thus condition (iii) is equivalent to $\ell_{[d}\ell_{e];i} = 0$, which
by Definition~\eqref{def:l-recur} means that $\ell$ is recurrent. But for
4-dimensional type~N spacetimes, the recurrence of the Weyl PND $\ell$
is equivalent to the complex recurrence of the self-dual Weyl tensor
$\bC^\dagger$ itself~\cite[\S 35.2]{kramer}. In turn, the complex
recurrent type~N vacua are precisely the vacuum
pp-waves~\cite{EhlersKundt}.
\end{proof}

\subsection{Classification and characterization} \label{sec:class-char}

We have now defined sufficiently many invariants to construct an IDEAL
characterization of the individual isometry classes of highly symmetric
pp-wave geometries. Our main results split into two parts. First, in
Theorem~\ref{thm:pp-wave-isom-classes}, we refine the the
classification of highly symmetric vacuum pp-wave geometries from
Proposition~\ref{prop:pp-wave-classes} down to individual isometry
classes. Each isometry class is labelled by a number of invariant
parameters and possibly a \emph{signature function} $F$. The signature
function appears as part of a functional relation between some scalar
invariants, which ultimately translate to an ODE on the function $A(u)$
or $\ff(z)$. These ODEs must have the property that the integration constants
in a general solution must exhaust the residual gauge freedom
(cf.~\eqref{eq:A-residual1}, \eqref{eq:A-residual2}, \eqref{eq:A-residual3}
and~\eqref{eq:A-residual4}) in the corresponding isometry class. The
differential order and size of the signature ODE will vary according to
the number of integration constants it needs to produce. The isometry
classes are conveniently listed in Table~\ref{tab:pp-wave-isom-classes}
and the logic of the proof follows the corresponding flowchart in
Figure~\ref{fig:pp-wave-classes}, which illustrates how different
classes reduce to each other by specializing parameters.

Second, in Theorem~\ref{thm:ideal-pp-wave-isom}, for each isometry
class, we give the IDEAL equations characterizing it. The
characterizations are conveniently listed in
Table~\ref{tab:pp-wave-ideal} and the logic of the proof is reflected in
the classification flowchart in Figure~\ref{fig:pp-wave-flowchart}.

\setlength{\LTcapwidth}{\textwidth}%
\begin{longtable}{lcc}
\caption{Isometry classes of highly symmetric vacuum pp-wave geometries.
	Any new notation is introduced in the proof of
	Theorem~\ref{thm:pp-wave-isom-classes} and can looked up by following
	the corresponding node number in Figure~\ref{fig:pp-wave-classes}.}
\label{tab:pp-wave-isom-classes}\\
	class & $f(\zeta,u)$ & invariant parameters \\
	\hline\hline
	\endfirsthead
	\caption[]{(continued)}\\
	class & $f(\zeta,u)$ & invariant parameters \\
	\hline\hline
	\endhead
	\noalign{\smallskip}
	$G_{6a;\alpha,\kappa}$:
		& $\frac{4\alpha}{u^2} u^{2i\kappa} \zeta^2$
		& $\alpha>0, \quad \kappa\ge 0$
	\\ \noalign{\smallskip}\hline\noalign{\smallskip}
	$G_{6b;\lambda}$:
		& $e^{2i\lambda u} \zeta^2$
		& $\lambda\ge 0$
	\\ \noalign{\smallskip}\hline\noalign{\smallskip}
	$G_{5';\alpha,F}$:
		& $\frac{4\alpha}{u^2} e^{i\Im B(u)} \zeta^2$
		& $\begin{gathered}
			\del_u(\Im\dot{B} e^{-\frac{\Re B}{2}}) \ne 0, 
			~~ \alpha > 0 , \\
			\Im \ddot{B} e^{-\Re B}
				= F(\Im \dot{B} e^{-\frac{\Re B}{2}}) ,
				\\
			F(-y) = -F(y) , ~~
			F(y) \ne -\frac{y}{2\sqrt{\alpha}} ,
				\\
			F\colon U \subset \mathbb{R} \to \mathbb{R}
			\end{gathered}$
	\\ \noalign{\smallskip}\hline\noalign{\smallskip}
	$G_{5';\oo,F}$:
		& $e^{i\Im B(u)} \zeta^2$
		& $\begin{gathered}
			\del_u(\Im\dot{B} e^{-\frac{\Re B}{2}}) \ne 0, \\
			\Im \ddot{B} e^{-\Re B}
				= F(\Im \dot{B} e^{-\frac{\Re B}{2}}) ,
				\\
			F(-y) = -F(y) , ~~
			F(y) \ne 0 ,
				\\
			F\colon U \subset \mathbb{R} \to \mathbb{R}
			\end{gathered}$
	\\ \noalign{\smallskip}\hline\noalign{\smallskip}
	$G_{5^\circ;F}$:
		& $e^{B(u)} \zeta^2$
		& \footnotesize$\begin{gathered}
			\del_u(\Re\dot{B} e^{-\frac{\Re B}{2}}) \ne 0 , \\
			\begin{pmatrix}
				\Re \ddot{B} e^{-\Re B} \\
				(\Im \dot{B})^2 e^{-\Re B}
			\end{pmatrix}
			= F(\Re \dot{B} e^{-\frac{\Re B}{2}}),
				\\
			F_{\Im}(y) \ge 0 ,
			~~ F_{\Re}(y) \ne \frac{1}{2} y^2 ,
				\\
			F = \begin{pmatrix} F_{\Re} \\ F_{\Im} \end{pmatrix}\colon U
				\subset \mathbb{R} \to \mathbb{R}^2
			\end{gathered}$
	\\ \noalign{\smallskip}\hline\noalign{\smallskip}
	$G_{3a';\pm\alpha}$:
		& $\pm\frac{4\alpha}{u^2} \ln \zeta$
		& $\{\pm\}, \quad \alpha > 0$
	\\ \noalign{\smallskip}\hline\noalign{\smallskip}
	$G_{3a^\circ;\alpha,\gamma}$:
		& $\frac{4\alpha}{u^2} e^{i\gamma} \ln \zeta$
		& $\gamma\in[0,\pi], \quad \alpha > 0$
	\\ \noalign{\smallskip}\hline\noalign{\smallskip}
	$G_{3b';\pm}$:
		& $\pm\ln \zeta$
		& $\{\pm\}$
	\\ \noalign{\smallskip}\hline\noalign{\smallskip}
	$G_{3b^\circ;\gamma}$:
		& $e^{i\gamma}\ln \zeta$
		& $\gamma\in[0,\pi]$
	\\ \noalign{\smallskip}\hline\noalign{\smallskip}
	$G_{3c;\lambda}$:
		& $e^{2\lambda \zeta}$
		& $\lambda > 0$
	\\ \noalign{\smallskip}\hline\noalign{\smallskip}
	$G_{3d;\gamma,\kappa}$:
		& $e^{i\gamma} \zeta^{2i\kappa}$
		& $\gamma \in [0,\pi] , \quad
			\kappa\in \mathbb{R} \setminus \{ 0 \}$
	\\ \noalign{\smallskip}\hline\noalign{\smallskip}
	$G_{2a^\circ;\kappa,c,\Re F}$:
		& $\frac{1}{u^2} \ff(\zeta u^{i\kappa})$
		& $\begin{gathered}
				\ff'''(z)^2/\ff''(z)^3 \ne \text{const} , \\
				\frac{\ff'''(z)}{\ff''(z)} = \exp F(\tfrac{\ff'''(z)^2}{\ff''(z)^3}) , \\
				\mathcal{F}''^{-1}(\ff''(z)) = Z_a, \\
				|\ff'(z) - \cF'(Z_a)| = c \ge 0 , \\
				F(y) \ne \frac{1}{2}\ln y + \text{const} , \\
				F\colon U \subset \mathbb{C} \to \mathbb{C}
			\end{gathered}$
	\\ \noalign{\smallskip}\hline\noalign{\smallskip}
	$G_{2a';\alpha,\gamma,\kappa,c}$:
		& \footnotesize$\frac{4\alpha}{u^2} e^{i\gamma} \left(\ln\zeta + c \zeta u^{i\kappa}\right)$
		& $\begin{gathered}
			\gamma \in [0,\pi], \quad \kappa \in \mathbb{R} , \\
			\alpha > 0 , \quad c > 0
			\end{gathered}$
	\\ \noalign{\smallskip}\hline\noalign{\smallskip}
	$G_{2b^\circ;c,\Re F}$:
		& $\ff(\zeta e^{iu})$
		& $\begin{gathered}
				\ff'''(z)^2/\ff''(z)^3 \ne \text{const} , \\
				\frac{\ff'''(z)}{\ff''(z)} = \exp F(\tfrac{\ff'''(z)^2}{\ff''(z)^3}) , \\
				\mathcal{F}''^{-1}(\ff''(z)) = Z_b, \\
				|\ff'(z) - \mathcal{F}'(Z_b)| = c \ge 0 , \\
				F(y) \ne \frac{1}{2}\ln y + \text{const} , \\
				F\colon U \subset \mathbb{C} \to \mathbb{C}
			\end{gathered}$
	\\ \noalign{\smallskip}\hline\noalign{\smallskip}
	$G_{2b';\gamma,\Re F}$:
		& $\ff(\zeta)$
		& \footnotesize$\begin{gathered}
				\ff'''(z) \ff'(z)/\ff''(z)^2 \ne \text{const} , \\
				\ff'''(z)/\ff''(z) = \exp F(\ff'''(z) \ff'(z)/\ff''(z)^2) , \\
				[\mathcal{F}'(-)/\mathcal{F}''(-)]^{-1}(\ff'(z)/\ff''(z)) = Z_{b'} , \\
				\Re\left[\frac{\ff''(z)}{\overline{\ff''(z)}}
					\frac{\overline{\cF''(Z_{b'})}}{\cF''(Z_{b'})}\right]^{1/2}
					= \cos\gamma \in [-1,1] , \\
				F(y) \ne \ln(\lambda y) , \ln\left(\frac{4c}{2-y}\right) , \\
				F\colon U \subset \mathbb{C} \to \mathbb{C}
			\end{gathered}$
	\\ \noalign{\smallskip}\hline\noalign{\smallskip}
	$G_{2b'';\gamma,\lambda,c}$:
		& \footnotesize$e^{i\gamma} \left(\ln\zeta + c \zeta e^{i\lambda u}\right)$
		& $\gamma \in [0,\pi], \quad \lambda \in \mathbb{R} , \quad c > 0$
	\\ \noalign{\smallskip}\hline\noalign{\smallskip}
	$G_{2b''';k}$:
		& $\zeta^{k}$
		& $k \in \mathbb{C} \setminus \left(i\mathbb{R} \cup \{1,2\}\right)$
	\\ \noalign{\smallskip}\hline\noalign{\smallskip}
	$G_{2c';\alpha,F}$:
		& $\frac{4\alpha}{u^2} e^{i \Im B(u)} \ln\zeta$
		& $\begin{gathered}
			\del_u(\Im\dot{B}) \ne 0 , 
			~~ \alpha > 0 , \\
			\Im \dot{B} e^{-\frac{\Re B}{2}}
				= F\left(\sin(\Im B)\right) ,
				\\
			F(-y) = -F(y) , ~~
			F(y) \ne 0 ,
				\\
			F\colon U \subset U(1) \to \mathbb{R}
			\end{gathered}$
	\\ \noalign{\smallskip}\hline\noalign{\smallskip}
	$G_{2c';\oo,F}$:
		& $e^{i \Im B(u)} \ln\zeta$
		& $\begin{gathered}
			\del_u(\Im\dot{B}) \ne 0 , \\
			\Im \dot{B} e^{-\frac{\Re B}{2}}
				= F\left(\sin (\Im B)\right) ,
				\\
			F(-y) = -F(y) , ~~
			F(y) \ne 0 ,
				\\
			F\colon U \subset U(1) \to \mathbb{R}
			\end{gathered}$
	\\ \noalign{\smallskip}\hline\noalign{\smallskip}
	$G_{2c^\circ;F}$:
		& $e^{B(u)} \ln\zeta$
		& $\begin{gathered}
			\del_u(\Re\dot{B} e^{-\frac{\Re B}{2}}) \ne 0 , \\
			\begin{pmatrix}
				\Re \ddot{B} e^{-\Re B} \\
				\cos(\Im B)
			\end{pmatrix}
				= F(\Re \dot{B} e^{-\frac{\Re B}{2}}) ,
				\\
			F_{\Re}(y) \ne \frac{1}{2} y^2 ,
				\\
			F=\begin{pmatrix} F_{\Re} \\ F_{\Im} \end{pmatrix}\colon U
				\subset \mathbb{R} \to \mathbb{R}\times [-1,1]
			\end{gathered}$
	\smallskip\\ \hline\hline
\end{longtable}

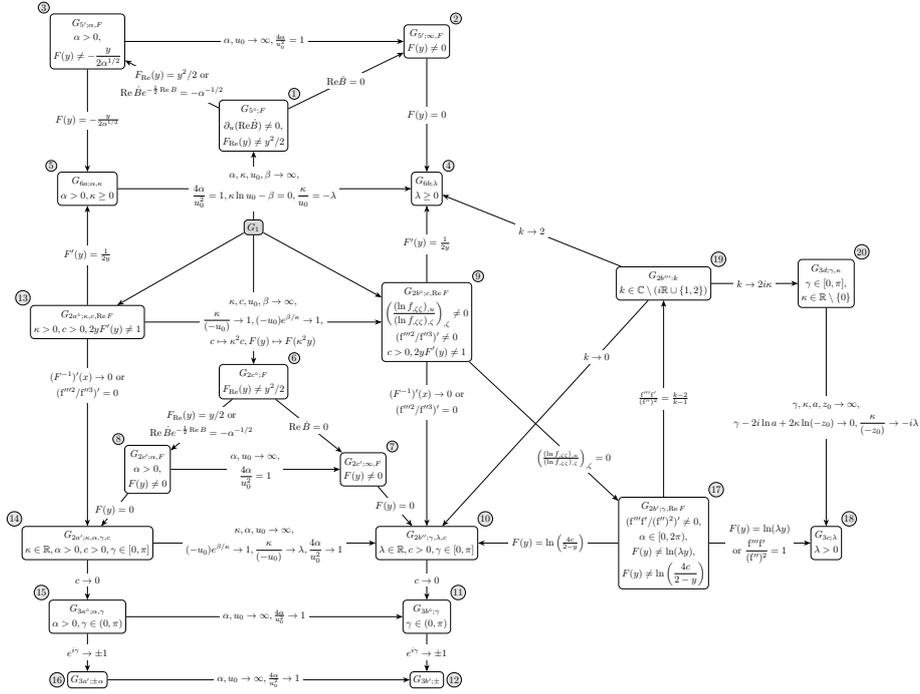
\begin{figure}

\tikzstyle{start}=[fill=white, draw=black, shape=trapezium, trapezium left angle=75, trapezium right angle=-75, very thick]
\tikzstyle{decide}=[fill=white, draw=black, shape=chamfered rectangle]
\tikzstyle{leaf}=[fill=white, draw=black, shape=rectangle, rounded corners, thick]
\tikzstyle{info}=[fill=white, draw=black, shape=rectangle]
\tikzstyle{yn}=[fill=white, shape=circle, inner sep=2pt]
\tikzstyle{num}=[fill={{gray!20}}, draw=black, shape=circle, inner sep=1pt, outer sep=3pt]

\tikzstyle{flow}=[->, >={Stealth[length=1.5ex]}, thick]
\tikzstyle{impl}=[->, >=stealth, double]
\tikzstyle{dd}=[-, dashed]
\tikzstyle{line}=[-, thick]
\resizebox{\textwidth}{!}{\begin{tikzpicture}
	\begin{pgfonlayer}{nodelayer}
		\node [style=leaf] (0) at (38.25, 23) {$\begin{gathered} G_{6b;\lambda} \\ \lambda\ge 0 \end{gathered}$};
		\node [style=num,above right] (26) at (0.north east) {4}; 
		\node [style=leaf] (1) at (26, 27.5) {$\begin{gathered} G_{5^\circ;F} \\ \del_u(\mathrm{Re}\dot{B}) \ne 0, \\ F_{\mathrm{Re}}(y) \ne y^2/2 \end{gathered}$};
		\node [style=num,above right] (21) at (1.north east) {1}; 
		\node [style=leaf] (2) at (14.25, 23) {$\begin{gathered} G_{6a;\alpha,\kappa} \\ \alpha>0, \kappa\ge 0 \end{gathered}$};
		\node [style=num,above left] (25) at (2.north west) {5}; 
		\node [style=leaf] (3) at (14.25, 33.5) {$\begin{gathered} G_{5';\alpha,F} \\ \alpha>0, \\ F(y)\ne -\frac{y}{2\alpha^{1/2}} \end{gathered}$};
		\node [style=num,above left] (23) at (3.north west) {3}; 
		\node [style=leaf] (4) at (38.25, -7.5) {$\begin{gathered} G_{3b^\circ;\gamma} \\ \gamma\in(0,\pi) \end{gathered}$};
		\node [style=num,above right] (30) at (4.north east) {11}; 
		\node [style=leaf] (5) at (14.25, -7.5) {$\begin{gathered} G_{3a^\circ;\alpha,\gamma} \\ \alpha>0, \gamma\in(0,\pi) \end{gathered}$};
		\node [style=num,above left] (29) at (5.north west) {15}; 
		\node [style=leaf] (6) at (26, 9.25) {$\begin{gathered} G_{2c^\circ;F} \\ F_{\mathrm{Re}}(y) \ne y^2/2 \end{gathered}$};
		\node [style=num,above right] (22) at (6.north east) {6}; 
		\node [style=leaf] (7) at (18.5, 3) {$\begin{gathered} G_{2c';\alpha,F} \\ \alpha>0, \\ F(y) \ne 0 \end{gathered}$};
		\node [style=num,above left] (27) at (7.north west) {8}; 
		\node [style=leaf] (8) at (38.25, -2.25) {$\begin{gathered} G_{2b'';\gamma, \lambda, c} \\ \lambda\in\mathbb{R}, c> 0, \gamma\in[0,\pi] \end{gathered}$};
		\node [style=num,above right] (32) at (8.north east) {10}; 
		\node [style=leaf, fill={gray!30}] (9) at (26, 20.25) {$G_1$};
		\node [style=leaf] (10) at (14.25, -2.25) {$\begin{gathered} G_{2a';\kappa,\alpha,\gamma,c} \\ \kappa\in\mathbb{R}, \alpha>0, c> 0, \gamma\in[0,\pi] \end{gathered}$};
		\node [style=num,above left] (31) at (10.north west) {14}; 
		\node [style=leaf] (11) at (66.5, 16.25) {$\begin{gathered} G_{3d;\gamma,\kappa} \\ \gamma\in[0,\pi], \\ \kappa\in\mathbb{R}\setminus\{0\} \end{gathered}$};
		\node [style=num,above right] (34) at (11.north east) {20}; 
		\node [style=leaf] (12) at (55, 16.25) {$\begin{gathered} G_{2b''':k} \\ k \in \mathbb{C} \setminus (i\mathbb{R} \cup \{1,2\}) \end{gathered}$};
		\node [style=num,above right] (35) at (12.north east) {19}; 
		\node [style=leaf] (13) at (14.25, 13.5) {$\begin{gathered}
		G_{2a^\circ;\kappa, c, \Re F} \\ \kappa>0, c>0, 2y F'(y) \ne 1  \end{gathered}$};
		\node [style=num,above left] (19) at (13.north west) {13}; 
		\node [style=leaf] (14) at (38.25, 13.5) {$\begin{gathered} G_{2b^\circ; c, \Re F} \\ \left(\frac{(\ln f_{,\zeta\zeta})_{,u}}{(\ln f_{,\zeta\zeta})_{,\zeta}}\right)_{,\zeta}\ne 0 \\ (\ff'''^2/\ff''^3)' \ne 0 \\ c>0, 2y F'(y) \ne 1 \end{gathered}$};
		\node [style=num,above right] (20) at (14.north east) {9}; 
		\node [style=leaf] (15) at (55, -2.25) {$\begin{gathered} G_{2b'; \gamma, \Re F} \\ (\ff''' \ff'/(\ff'')^2)' \ne 0, \\ \alpha\in[0,2\pi), \\ F(y) \ne \ln(\lambda y), \\ F(y) \ne \ln\left(\frac{4c}{2-y}\right) \end{gathered}$};
		\node [style=num,above right] (36) at (15.north east) {17}; 
		\node [style=leaf] (16) at (66.5, -2.25) {$\begin{gathered} G_{3c;\lambda} \\ \lambda>0 \end{gathered}$};
		\node [style=num,above right] (33) at (16.north east) {18}; 
		\node [style=leaf] (17) at (38.25, 33.5) {$\begin{gathered} G_{5';\infty,F} \\ F(y) \ne 0 \end{gathered}$};
		\node [style=num,above right] (24) at (17.north east) {2}; 
		\node [style=leaf] (18) at (33.75, 3) {$\begin{gathered} G_{2c';\infty,F} \\ F(y) \ne 0 \end{gathered}$};
		\node [style=num,above right] (28) at (18.north east) {7}; 
		\node [style=leaf] (37) at (14.25, -12) {$G_{3a';\pm\alpha}$};
		\node [style=num,left] (39) at (37.west) {16}; 
		\node [style=leaf] (38) at (38.25, -12) {$G_{3b';\pm}$};
		\node [style=num,right] (40) at (38.east) {12}; 
	\end{pgfonlayer}
	\begin{pgfonlayer}{edgelayer}
		\draw [style=flow] (9) to (1);
		\draw [style=flow] (9) to (13);
		\draw [style=flow] (9) to (6);
		\draw [style=flow] (9) to (14);
		\draw [style=flow] (13) to node [fill=white, midway] {$\begin{gathered} \kappa,c,u_0,\beta\to\infty, \\ \frac{\kappa}{(-u_0)} \to 1, (-u_0) e^{\beta/\kappa} \to 1, \\ c \mapsto \kappa^2 c, F(y) \mapsto F(\kappa^2 y) \end{gathered}$} (14);
		\draw [style=flow] (13) to node [fill=white, pos=0.25] {$\begin{gathered} (F^{-1})'(x)\to 0 \text{ or}\\ (\ff'''^2/\ff''^3)' = 0 \end{gathered}$} (10);
		\draw [style=flow] (14) to node [fill=white, pos=0.7] {$\left(\frac{(\ln f_{,\zeta\zeta})_{,u}}{(\ln f_{,\zeta\zeta})_{,\zeta}}\right)_{,\zeta} = 0$} (15);
		\draw [style=flow] (10) to node [fill=white, midway] {$\begin{gathered} \kappa,\alpha,u_0\to\infty, \\ (-u_0) e^{\beta/\kappa} \to 1, \frac{\kappa}{(-u_0)} \to \lambda, \frac{4\alpha}{u_0^2}\to 1 \end{gathered}$} (8);
		\draw [style=flow] (14) to node [fill=white, pos=0.25] {$\begin{gathered} (F^{-1})'(x) \to 0 \text{ or}\\ (\ff'''^2/\ff''^3)' = 0 \end{gathered}$} (8);
		\draw [style=flow] (10) to node [fill=white, midway] {$c\to 0$} (5);
		\draw [style=flow] (5) to node [fill=white, midway] {$\alpha,u_0\to\infty, \frac{4\alpha}{u_0^2} \to 1$} (4);
		\draw [style=flow] (8) to node [fill=white, midway] {$c\to 0$} (4);
		\draw [style=flow] (13) to node [fill=white, midway] {$F'(y) = \frac{1}{2y}$} (2);
		\draw [style=flow] (2) to node [fill=white, midway] {$\begin{gathered} \alpha,\kappa,u_0,\beta \to \infty, \\ \frac{4\alpha}{u_0^2} = 1, \kappa\ln u_0-\beta = 0, \frac{\kappa}{u_0} = -\lambda \end{gathered}$} (0);
		\draw [style=flow] (14) to node [fill=white, midway] {$F'(y) = \frac{1}{2y}$} (0);
		\draw [style=flow] (3) to node [fill=white, midway] {$F(y) = -\frac{y}{2\alpha^{1/2}}$} (2);
		\draw [style=flow] (17) to node [fill=white, midway] {$F(y) = 0$} (0);
		\draw [style=flow] (1) to node [fill=white, midway] {$\begin{gathered} F_{\mathrm{Re}}(y) = y^2/2 \text{ or}\\ \Re\dot{B} e^{-\frac{1}{2} \Re B} = -\alpha^{-1/2} \end{gathered}$} (3);
		\draw [style=flow] (11) to node [fill=white, midway] {$\begin{gathered} \gamma,\kappa,a,z_0 \to \infty,\\ \gamma - 2i\ln a + 2\kappa \ln (-z_0) \to 0, \frac{\kappa}{(-z_0)} \to -i\lambda \end{gathered}$} (16);
		\draw [style=flow] (15) to node [fill=white, midway] {$\begin{gathered} F(y) = \ln(\lambda y)\\ \text{or } \frac{\ff''' \ff'}{(\ff'')^2} = 1 \end{gathered}$} (16);
		\draw [style=flow] (12) to node [fill=white, midway] {$k \to 2i\kappa$} (11);
		\draw [style=flow] (15) to node [fill=white, midway] {$\frac{\ff''' \ff'}{(\ff'')^2} = \frac{k-2}{k-1}$} (12);
		\draw [style=flow] (6) to node [fill=white, midway] {$\begin{gathered} F_{\mathrm{Re}}(y) = y/2 \text{ or}\\ \Re\dot{B} e^{-\frac{1}{2} \Re B} = -\alpha^{-1/2} \end{gathered}$} (7);
		\draw [style=flow] (7) to node [fill=white, midway] {$\begin{gathered} \alpha,u_0\to \oo, \\ \frac{4\alpha}{u_0^2} = 1 \end{gathered}$} (18);
		\draw [style=flow] (15) to node [fill=white, midway] {$F(y) = \ln\left(\frac{4c}{2-y}\right)$} (8);
		\draw [style=flow] (1) to node [fill=white, midway] {$\mathrm{Re}\dot{B} = 0$} (17);
		\draw [style=flow] (12) to node [fill=white, midway] {$k \to 2$} (0);
		\draw [style=flow] (7) to node [fill=white, midway] {$F(y)=0$} (10);
		\draw [style=flow] (18) to node [fill=white, midway] {$F(y)=0$} (8);
		\draw [style=flow] (3) to node [fill=white, midway] {$\alpha,u_0\to \oo, \frac{4\alpha}{u_0^2} = 1$} (17);
		\draw [style=flow] (6) to node [fill=white, midway] {$\Re\dot{B} = 0$} (18);
		\draw [style=flow] (12) to node [fill=white, pos=0.25] {$k \to 0$} (8);
		\draw [style=flow] (5) to node [fill=white, midway] {$e^{i\gamma} \to \pm 1$} (37);
		\draw [style=flow] (4) to node [fill=white, midway] {$e^{i\gamma} \to \pm 1$} (38);
		\draw [style=flow] (37) to node [fill=white, midway] {$\alpha,u_0\to\infty, \frac{4\alpha}{u_0^2} \to 1$} (38);
	\end{pgfonlayer}
\end{tikzpicture}}
\caption{Flowchart for relations between pp-wave classes. Arrows
indicate transitions from generic to more special classes by taking
limits of invariant or gauge parameters. Numbers accompany the proof of
Theorem~\ref{thm:pp-wave-isom-classes}.}
\label{fig:pp-wave-classes}
\end{figure}

\begin{thm}[pp-wave isometry classes] \label{thm:pp-wave-isom-classes}
Consider a point $x\in M$, where $(M,g)$ belongs to one of the highly
symmetric regular (Definition~\ref{def:regular}) vacuum pp-wave
geometries from Proposition~\ref{prop:pp-wave-classes}. Then there exists a neighborhood of $x$ such
that this neighborhood is isometric to exactly one of the classes listed
in Table~\ref{tab:pp-wave-isom-classes} and
Figure~\ref{fig:pp-wave-classes}. These classes are identified by
invariant parameters and/or a signature function $F$ defined on a domain
$U$. The domain $U$ needs to be chosen to correspond to the needed
neighborhood of $x$. Limiting values of the invariant parameters allow
for transitions between different classes according to the flowchart in
Figure~\ref{fig:pp-wave-classes}.
\end{thm}

\begin{proof}
We will proceed by going through the flowchart in
Figure~\ref{fig:pp-wave-classes} step by step, in the order of the
numbers indicated next to each node. In the corresponding numbered part
of the proof, we give some invariant conditions that select a subclass
of formfactor functions $f(\zeta,u)$. These subclasses are then broken
up into individual isometry classes under some non-degeneracy
conditions. The isometry classes and their invariant parameter
corresponding to that node in Figure~\ref{fig:pp-wave-classes} are
precisely recorded by representative formfactors in
Table~\ref{tab:pp-wave-isom-classes}. Then, to be exhaustive, the ways
in which the formfactors $f(\zeta,u)$ fail to satisfy the non-degeneracy
conditons, if any, are broken up into cases, which correspond to the
arrows leaving the given node in Figure~\ref{fig:pp-wave-classes} and
are discussed separately. This way, we show that all of the highly
symmetric formfactor families listed in~\eqref{eq:f-cases} are exhausted
by our analysis.

The functional forms of $f(u,\zeta)$ listed in~\eqref{eq:f-cases} are
not all mutually exclusive and each one may contain multiple
individual isometry classes within it. Besides $G_1$, the most generic subclasses
are $G_{2a}$, $G_{2b}$, $G_{2c}$ and $G_5$, each containing functional
degrees of freedom. The remaining subclasses will be specializations
thereof. Invariant parameters that label an isometry class will appear
in relations among differential invariants (\emph{classifying} or
\emph{signature} relations). We know that we have isolated an individual
isometry class when the integration parameters of the corresponding set
of classifying relations can be completely absorbed by residual gauge
freedom. A combination of $f(\zeta,u)$ and its derivatives can be
checked to be invariant under residual gauge freedom by direct
inspection or, as a shortcut, by realizing it as a component of one of
the tensor or conditional curvature invariants that we have defined in
Section~\ref{sec:pp-wave}.

To start, note that $f_{,\zeta\zeta} = 0$ is equivalent to $\bC^\dagger
= 0$~\eqref{eq:bC-expr}, which together with the vacuum condition
implies that the spacetime is locally Minkowski. Hence we can presume
below that $f_{,\zeta\zeta}\ne 0$.

\begin{enumerate}

\item \label{itm:K2=0} 
The invariant relation
\begin{equation} \label{eq:K2=0}
	(\ln f_{,\zeta\zeta})_{,\zeta}
	= \frac{f_{,\zeta\zeta\zeta}}{f_{,\zeta\zeta}} = 0
\end{equation}
integrates to $f(\zeta,u) = A(u) \zeta^2 + A_1(u) \zeta + A_0(u)$, where
we can set $A_1(u) = A_0(u)$ by gauge freedom, which lands us in the
$G_5$ subclass. Rewriting $A(u) = e^{B(u)}$, the combinations $\dot{B}
e^{-\frac{1}{2}\Re B}$ and $\ddot{B} e^{-\Re B}$ are invariant under the
continuous gauge symmetries. The discrete gauge symmetry $\zeta \mapsto
\bar{\zeta}$ changes the signs of $\Im \dot{B} e^{-\frac{1}{2}\Re B}$
and $\Im\ddot{B} e^{-\Re B}$, while $u\mapsto -u$ changes the sign of
$\dot{B} e^{-\frac{1}{2}\Re B}$. Generically, $\del_u(\Re \dot{B}
e^{-\frac{1}{2}\Re B}) \ne 0$, so there exists a signature function
$F=(F_{\Re},F_{\Im})\colon U\subset \mathbb{R} \to \mathbb{R}^2$, with
$F_{\Im}(y) \ge 0$, such that
\begin{equation} \label{eq:G5-signature}
	\begin{pmatrix}
		\Re \ddot{B} e^{-\Re B} \\ (\Im \dot{B})^2 e^{-\Re B}
	\end{pmatrix}
	= F\left(\Re \dot{B} e^{-\frac{1}{2}\Re B}\right) .
\end{equation}
Given any particular solution $B_F(u)$, symmetries dictate that the
general 3-parameter solution is
\begin{equation}
	B(u) = B_F(a (u-u_0)) +  2\ln a + 2i\beta
		\quad \text{or} \quad
	\overline{B_F(a (u-u_0))} +  2\ln a - 2i\beta ,
\end{equation}
where the integration constants $a, \beta, u_0$ and the complex
conjugaton can be absorbed by corresponding gauge transformations. Hence
the signature function $F$ completely describes the corresponding
isometry class, which we label by $G_{5^\circ;F}$. The choice $F_{\Re}(y)
= y^2/2$ implies $\Re B = -2\ln|u-u_0| + \ln 4\alpha$, which implies
constant $\Re\dot{B} e^{-\frac{1}{2} \Re B} = \pm\alpha^{-1/2}$, which
must be excluded to be consistent with our genericity assumption.

The only possibilities left to consider are those with constant
$\Re\dot{B} e^{-\frac{1}{2} \Re B} \ne 0$ and $\Re\dot{B}
e^{-\frac{1}{2} \Re B} = 0$.

\item \label{itm:G5-ReIa=0} 
The condition
\begin{equation} \label{eq:G5-ReIa=0}
	\Re\dot{B} e^{-\frac{1}{2} \Re B} = 0
\end{equation}
integrates to a constant $\Re B(u) = 2\ln a$. Generically then,
$\del_u(\Im \dot{B} e^{-\frac{1}{2}\Re B}) \ne 0$, so there exists a
signature function $F\colon U\subset \mathbb{R} \to \mathbb{R}$, with
$F(-y) = -F(y)$ because both $\Im B$ and $-\Im B$ must be solutions, such
that
\begin{equation} \label{eq:G5-ReIa=0-signature}
	\Im \ddot{B} e^{-\Re B}
	= F\left(\Im \dot{B} e^{-\frac{1}{2}\Re B}\right) .
\end{equation}
Given any particular solution $\Im B_F(u)$, symmetries dictate that the
general 2-parameter solution is
\begin{equation}
	B(u) = i\Im B_F(a(u-u_0)) + 2\ln a + 2i\beta
		\quad \text{or} \quad
	-i\Im B_F(a(u-u_0)) + 2\ln a - 2i\beta .
\end{equation}
All the integration constants $a, \beta, u_0$ and the complex
conjugation can be absorbed by corresponding gauge transformations.
Hence, the signature function $F$ completely describes the corresponding
isometry class, which we label by $G_{5';\oo,F}$. The choice $F(y) = 0$
implies $\Im B = 2a\lambda (u-u_0)$, which implies constant
$\Im\dot{B}e^{-\frac{1}{2}\Re B} = 2\lambda$, which must be excluded to
be consistent with our genericity assumption.

The only possibility left to consider is that with constant $\Im\dot{B}
e^{-\frac{1}{2} \Re B}$.

\item \label{itm:G5-ReIa=const} 
Let us parametrize
\begin{equation} \label{eq:G5-ReIa=const}
	\Re\dot{B} e^{-\frac{1}{2} \Re B} = \pm\alpha^{-1/2}
\end{equation}
with constant $\alpha > 0$, which integrates to $\Re B(u) = -2\ln|u-u_0|
+ \ln 4\alpha$ with integration constant $u_0$, with the $+$ and $-$
signs corresponding to $u<u_0$ and $u>u_0$ respectively. Generically then,
$\del_u(\Im \dot{B} e^{-\frac{1}{2}\Re B}) \ne 0$, so there exists a
signature function $F\colon U\subset \mathbb{R} \to \mathbb{R}$, with
$F(-y) = -F(y)$ because both $\Im B$ and $-\Im B$ must be solutions, such
that
\begin{equation}
	\Im \ddot{B} e^{-\Re B}
	= F\left(\Im \dot{B} e^{-\frac{1}{2}\Re B}\right) .
\end{equation}
Given any particular solution $\Im B_F(u-u_0)$, symmetries dictate that the
general 2-parameter solution is
\begin{multline}
	B(u) = i\Im B_F(a(u-u_0)) - 2\ln|u-u_0| + \ln4\alpha + 2i\beta
		\\ \text{or} \quad
	-i\Im B_F(a(u-u_0)) - 2\ln|u-u_0| + \ln4\alpha - 2i\beta .
\end{multline}
All the integration constants $a, \beta, u_0$ and the complex
conjugation can be absorbed by corresponding gauge transformations.
Hence, the signature function $F$ completely describes the corresponding
isometry class, which we label by $G_{5';\alpha,F}$. The choice $F(y) =
-y/(2\alpha^{1/2})$ implies $\Im B = 2\kappa\ln|u-u_0|$, which implies
constant $\Im\dot{B}e^{-\frac{1}{2}\Re B} = \kappa \alpha^{-1/2}$, which
must be excluded to be consistent with our genericity assumption.

The only possibility left to consider is that with constant $\Im\dot{B}
e^{-\frac{1}{2} \Re B}$.

Taking the invariant parameter limit $\alpha \to \oo$, accompanied by
the gauge parameter limit $u_0\to \oo$, such that $ 4\alpha/u_0^2 \to
1$, we recover the $G_{5';\oo,F}$ subclass.

\item \label{itm:G6b} 
The condition
\begin{equation}
	\dot{B} e^{-\frac{1}{2} \Re B} = 2i\lambda
	\iff
	\Re\dot{B} e^{-\frac{1}{2} \Re B} = 0, \quad
	\Im\dot{B} e^{-\frac{1}{2} \Re B} = 2\lambda
\end{equation}
integrates to $B(u) = 2\ln a + 2ia\lambda |u-u_0|$ with integration
constants $a$, $u_0$, which can be absorbed by the corresponding
continuous gauge freedom. Under $\zeta \mapsto \bar{\zeta}$, $\lambda
\mapsto -\lambda$, so we can choose to represent the isometry class
with $\lambda \ge 0$. Hence, the above condition already completely
describes the corresponding isometry class, which we label by
$G_{6b;\lambda}$.

\item \label{itm:G6a} 
The condition
\begin{multline}
	\dot{B} e^{-\frac{1}{2} \Re B} = (i\kappa \pm 1) \alpha^{-1/2}
	\\ \iff
	\Re\dot{B} e^{-\frac{1}{2} \Re B} = \pm\alpha^{-1/2}, \quad
	\Im\dot{B} e^{-\frac{1}{2} \Re B} = \kappa \alpha^{-1/2} ,
\end{multline}
with $\alpha > 0$, integrates to
\begin{equation}
	B(u) = i\beta + \ln4\alpha + 2(i\kappa \sgn(u-u_0) - 1)\ln|u-u_0|
\end{equation}
with integration constants $\beta$, $u_0$, which can be absorbed by the
corresponding continuous gauge freedom. The $+$ and $-$ signs correspond
to $u<u_0$ and $u>u_0$ respectively, with the differrence between the
two cases absorbed by the $u\mapsto -u$ residual gauge transformation,
after fixing $u_0=0$. Without loss of generality, we can then assume
that $u>0$. Under $\zeta \mapsto \bar{\zeta}$, $\kappa \mapsto -\kappa$,
so we can choose to represent the isometry class with $\kappa \ge 0$.
Hence, the above condition already completely describes the
corresponding isometry class, which we label by $G_{6a;\alpha,\kappa}$.

Taking the invariant parameter limit $\alpha,\kappa \to \oo$,
accompanied by the gauge parameter limit $\beta, u_0\to \oo$, such that
$\frac{\kappa}{2\alpha^{1/2}} = \lambda$, $4\alpha/u_0^2 \to 1$, and
$\kappa\ln u_0 - \beta = 0$, we recover the $G_{6b;\lambda}$ subclass.

\item \label{itm:G2c} 
The invariant relation
\begin{equation} \label{eq:log-condition}
	-\left(\frac{1}{(\ln f_{,\zeta\zeta})_{,\zeta}}\right)_{,\zeta}
	= \frac{1}{2} ,
\end{equation}
where the invariance of the left-hand side is checked by noting that it
is equal to the $\Ic$ conditional curvature invariant in~\eqref{eq:Ic-def},
integrates to $f(\zeta,u) = A(u) \ln(\zeta - h(u)) + A_1(u) \zeta +
A_0(u)$, where we can set $h(u) = A_0(u) = 0$ by gauge freedom. For
this form of $f(\zeta,u)$, using the invariant $N$~\eqref{eq:N-def}
and its evaluation in~\eqref{eq:N-eval-G2c}, the invariant condition
$N=0$ is equivalent to
\begin{equation} \label{eq:G2c-A1=0}
	A_1(u) = 0 ,
\end{equation}
which lands us in the $G_{2c}$ subclass.

Once again, rewriting $A(u) = e^{B(u)}$, the combinations $\Re \dot{B}
e^{-\frac{1}{2}\Re B}$, $\Re\ddot{B} e^{-\Re B}$ and $\Im B$ are
invariant under the continuous gauge symmetries. The discrete gauge
symmetry $u \mapsto -u$ changes the signs of $\Re \dot{B}
e^{-\frac{1}{2}\Re B}$ and $\Im\ddot{B} e^{-\Re B}$, while $\zeta
\mapsto \bar{\zeta}$ changes the sign of $\Im B$. Generically,
$\del_u(\Re \dot{B} e^{-\frac{1}{2}\Re B}) \ne 0$, so there exists a
signature function $F=(F_{\Re},F_{\Im})\colon U\subset \mathbb{R} \to
\mathbb{R}\times [-1,1]$, such that
\begin{equation} \label{eq:G2c-signature}
	\begin{pmatrix}
		\Re \ddot{B} e^{-\Re B} \\ \cos(\Im B)
	\end{pmatrix}
	= F\left(\Re \dot{B} e^{-\frac{1}{2}\Re B}\right) .
\end{equation}
Given any particular solution $B_F(u)$, symmetries dictate that the
general 2-parameter solution is
\begin{equation}
	B(u) = B_F(a (u-u_0)) +  2\ln a
		\quad \text{or} \quad
	\overline{B_F(a (u-u_0))} +  2\ln a ,
\end{equation}
where the integration constants $a, u_0$ and the complex conjugaton can
be absorbed by corresponding gauge transformations. Hence the signature
function $F$ completely describes the corresponding isometry class,
which we label by $G_{2^\circ;F}$. The choice $F_{\Re}(y) = y^2/2$
implies $\Re B = -2\ln|u-u_0| + \ln 4\alpha$, which implies constant
$\Re\dot{B} e^{-\frac{1}{2} \Re B} = \pm\alpha^{-1/2}$, which must be
excluded to be consistent with our genericity assumption.

The only possibilities left to consider are those with constant
$\Re\dot{B} e^{-\frac{1}{2} \Re B} \ne 0$ and $\Re\dot{B} = 0$.

\item \label{itm:G2c0-ReIa=0} 
The condition
\begin{equation}
	\Re \dot{B} = 0
\end{equation}
integrates to a constant $\Re B(u) = 2\ln a$. Generically then,
$\del_u(\Im B) \ne 0$, so there exists a signature function $F\colon
U\subset [-1,1] \to \mathbb{R}$, because both $\Im B$ and $-\Im B$ must
be solutions, with $F(-y) = -F(y)$ such that
\begin{equation} \label{eq:G2c0-ReIa=0-signature}
	\Im \dot{B} e^{-\frac{1}{2}\Re B}
	= F\left(\sin(\Im B)\right) .
\end{equation}
Given any particular solution $\Im B_F(u)$, symmetries dictate that the
general 2-parameter solution is
\begin{equation}
	B(u) = i\Im B_F(a(u-u_0)) + 2\ln a
		\quad \text{or} \quad
	-i\Im B_F(a(u-u_0)) + 2\ln a ,
\end{equation}
where the integration constants $a, u_0$ and the complex conjugation can
be absorbed by corresponding gauge transformations. Hence, the signature
function $F$ completely describes the corresponding isometry class,
which we label by $G_{2c';\oo,F}$. The choice $F(y) = 0$ implies
constant $\Im B(u) = \gamma$, which must be excluded to be consistent
with our genericity assumption.

The only possibility left to consider is constant $\Im B$.

\item \label{itm:G2c0-ReIa=const} 
Let us parametrize
\begin{equation}
	\Re\dot{B} e^{-\frac{1}{2} \Re B} = \pm\alpha^{-1/2}
\end{equation}
with constant $\alpha > 0$, which integrates to $\Re B(u) = -2\ln|u-u_0|
+ \ln 4\alpha$ with integration constant $u_0$, with the $+$ and $-$
signs corresponding to $u<u_0$ and $u>u_0$ respectively. Generically
then, $\del_u(\Im B) \ne 0$, so there exists a signature function
$F\colon U\subset \mathbb{R} \to \mathbb{R}$, with $F(-y) = -F(y)$
because both $\Im B$ and $-\Im B$ must be solutions, such that
\begin{equation}
	\Im \dot{B} e^{-\Re B}
	= F\left(\sin(\Im B)\right) .
\end{equation}
Given any particular solution $\Im B_F(u-u_0)$, symmetries dictate that the
general 2-parameter solution is
\begin{multline}
	B(u) = i\Im B_F(a(u-u_0)) - 2\ln|u-u_0| + \ln4\alpha
		\\ \text{or} \quad
	-i\Im B_F(a(u-u_0)) - 2\ln|u-u_0| + \ln4\alpha ,
\end{multline}
where the integration constants $a, u_0$ and the complex conjugation can
be absorbed by corresponding gauge transformations. The difference
between $u<u_0$ and $u>u_0$ can be absorbed by the $u\mapsto -u$
residual gauge transformation after fixing $u_0=0$. Hence, the signature
function $F$ completely describes the corresponding isometry class,
which we label by $G_{2c';\alpha,F}$. The choice $F(y) = 0$ implies
constant $\Im B = \gamma$,  which must be excluded to be consistent with
our genericity assumption.

The only possibility left to consider is that with constant $\Im B$.

Taking the invariant parameter limit $\alpha \to \oo$, accompanied by
the gauge parameter limit $u_0\to \oo$, such that $ 4\alpha/u_0^2 \to
1$, we recover the $G_{2c';\oo,F}$ subclass.

\item \label{itm:G2b} 
Consider the equation
\begin{equation} \label{eq:G2b-entry}
	\left(\frac{(\ln f_{,\zeta\zeta})_{,u}}{(\ln f_{,\zeta\zeta})_{,\zeta}}\right)_{\zeta}
	= i\lambda \left[\left(\frac{2}{(\ln f_{,\zeta\zeta})_{,\zeta}}\right)_{\zeta}
		+ 1\right] .
\end{equation}
The bracketed expression on the right-hand side is an invariant, while
the left-hand side is invariant up to changing by a factor of $a$ under
$u\mapsto u/a$. So to make the equation invariant, the parameter
$\lambda$ must also transform as $\lambda \mapsto a\lambda$. Thus,
$\lambda \in \mathbb{R}$ is not an invariant parameter, but the property
$\lambda = 0$ or $\lambda \ne 0$ is invariant. The first integral of the
above equation is
\begin{multline}
	\frac{(\ln f_{,\zeta\zeta})_{,u}}{(\ln f_{,\zeta\zeta})_{,\zeta}}
	- \frac{2i\lambda}{(\ln f_{,\zeta\zeta})_{,\zeta}}
	= i\lambda\zeta - e^{-i\lambda u} h'(u)
	\\ \iff
	e^{i\lambda u} (\ln e^{-2i\lambda u} f_{,\zeta\zeta})_{,u}
	= (\zeta i\lambda e^{i\lambda u} - h'(u))
	(\ln e^{-2i\lambda u} f_{,\zeta\zeta})_{,\zeta}
	\\ \iff
	\left[(\zeta e^{i\lambda u} - h(u))_{,\zeta} \del_u
	- (\zeta e^{i\lambda u} - h(u))_{,u} \del_\zeta\right]
	\ln e^{-2i\lambda u} f_{,\zeta\zeta} = 0 ,
\end{multline}
with implicitly defined integration parameter $h(u)$. The general
solution to the last equation is clearly
\begin{multline}
	\ln e^{-2i\lambda u} f_{,\zeta\zeta} = \ln \ff''(\zeta e^{i\lambda u} - h(u))
	\\ \iff
	f(\zeta,u) = \ff(\zeta e^{i\lambda u} - h(u)) + A_1(u) \zeta + A_0(u) ,
\end{multline}
for some holomorphic function $\ff(z)$ and integration parameters
$A_0(u)$, $A_1(u)$. We can set $h(u) = A_0(u) = 0$ by gauge freedom. For
this form of $f(\zeta,u)$, the value of the conditional curvature invariant
$\bM^{(\oo)}$ is given in~\eqref{eq:M-eval-G2b}.
So setting $\bM^{(\oo)} = 0$ forces $\ff'''(z) = 0$ or $A_1(u) \zeta =
c \zeta e^{i\lambda u}$. The former possibility results in a
formfactor corresponding to the $G_5$ subclass, while in the latter
possibility the term $A_1(u) \zeta$ can be absorbed into the definition
of $\ff(z)$, which lands us in the $G_{2b}$ subclass.

In this branch of the classification, we assume now that $\ff'''(z) \ne
0$, and the above conditions hold with $\lambda \in \mathbb{R}$, so that
we are in the $G_{2b}$ subclass.

Suppose that
\begin{equation} \label{eq:G2b-generic-ineq}
	\left(\frac{(\ln f_{,\zeta\zeta})_{,u}}
		{(\ln f_{,\zeta\zeta})_{,\zeta}}\right)_{\zeta} \ne 0 .
\end{equation}
We consider separately the subfamily where instead the above holds as an
equality, which corresponds to either $\lambda = 0$ or $-\left(1/(\ln
f_{,\zeta\zeta})_{,\zeta}\right) = \frac{1}{2}$.

Comparing~\eqref{eq:G2b-entry} to the definition of the invariant
$\Xi^{(\oo)}$ in~\eqref{eq:Xioo-def}, we see that it implies that the
invariant $\Xi^{(\oo)} = i\lambda$, while the
inequality~\eqref{eq:G2b-generic-ineq} implies $\lambda \ne 0$, hence we
can set $\lambda = 1$ by setting the gauge parameter $a=1/\lambda$,
which we assume from now on. The residual gauge transformations that
preserve the formfactor $f(\zeta,u) = \ff(\zeta e^{iu})$ are $\zeta
\mapsto e^{i\beta}\zeta$, which induces $\ff(z) \mapsto \ff(e^{-i\beta}
z)$, $\zeta \mapsto \zeta + z_0 e^{-iu}$, which induces $\ff(z) \mapsto
\ff(z-z_0) - \overline{z_0} z + \frac{1}{2} |z_0|^2$, $v\mapsto v-2u\Re
w_0$ with $c(u)=\Im w_0$, which induces $\ff(z) \mapsto \ff(z) + w_0$,
and $\zeta \mapsto \bar{\zeta}$, which induces $\ff(z) \mapsto
\bar{\ff}(z)$.

Generically, $\del_z (\ff'''(z)^2/\ff''(z)^3) \ne 0$, so then there
exists a holomorphic function $F\colon U \subset \mathbb{C} \to
\mathbb{C}$ such that
\begin{equation} \label{eq:G2b-signature}
	\ln\frac{\ff'''(z)}{\ff''(z)}
	= F\left(\frac{\ff'''(z)^2}{\ff''(z)^3}\right) .
\end{equation}
The argument of $F(-)$ is an invariant because it coincides with the
value of the invariant $\Id^{(\oo)}$~\eqref{eq:Idoo-def}, which was
evaluated in~\eqref{eq:Ide-eval-G2b}. Neither the left-hand side
in~\eqref{eq:G2b-signature} nor $F(y)$ itself are entirely invariant, as
$\ff(z) \mapsto \ff(e^{-i\beta} z)$ induces $F(y) \mapsto F(y) - i\beta$
and $\ff(z) \mapsto \bar{\ff}(z)$ induces $F(y) \mapsto \bar{F}(y)$.
However, $\Re F(y)$ is invariant under these transformations and also
uniquely fixes $F(y)$ up to these ambiguities. Thus, we consider $\Re
F\colon U \subset \mathbb{C} \to \mathbb{R}$ as our signature function.

If $F'(y) = 1/(2y)$ or equivalently $F(y) = \frac{1}{2}\ln (2A y)$, then
the ODE~\eqref{eq:G2b-signature} no longer depends on $\ff'''(z)$ and
reduces to $\ff''(z) = 2A$, and hence $\ff'''(z) = 0$, which contradicts
our earlier hypothesis. Otherwise, the solution $\ff(z) = A z^2 + A_1 z
+ A_0$ recovers the $G_{6b}$ subclass.

Provided $F'(y) \ne 1/(2y)$, the ODE~\eqref{eq:G2b-signature} can be
solved for $\ff'''(z)$, so the usual existence and uniqueness theory
holds for it. If $\ff(z) = \cF(z)$ is a particular solution, then
symmetries dictate that the general 3-parameter solution is
\begin{equation}
	\ff(z) = \cF(z-z_0) + A_{1,0} (z-z_0) + A_{0,0} .
\end{equation}
While the $z_0$ and $A_{0,0}$ integration constants can be absorbed by
the residual gauge freedom, the parameter $A_{1,0}$ cannot be, and it
does contain some invariant information.
To be concrete, we uniquely fix%
	\footnote{There may be other ways to uniquely fixing the solution
	$\cF(z)$, for instance by different initial conditions at $z=0$ or by
	asyptotic conditions at $z=\oo$.} %
$\cF(z)$ from $F(y)$ by requiring that $\cF(z) = \frac{1}{2} z^2 +
O(z^3)$. The gauge transformation $F(y) \mapsto F(y) - i\beta$ induces
$\cF(z) \mapsto \cF(e^{-i\beta} z)$, $A_{1,0} \mapsto A_{1,0}
e^{-i\beta}$, and $F(y) \mapsto \bar{F}(y)$ induces $\cF(z) \mapsto
\overline{\cF}(z)$, $A_{1,0} \mapsto \overline{A_{1,0}}$. So only
$|A_{1,0}|$ is invariant in $A_{1,0}$. Using the
ODE~\eqref{eq:G2b-signature}, and computing the invariants $\Id^{(\oo)}$
and $\Ie^{(\oo)}$, from~\eqref{eq:Idoo-def} and~\eqref{eq:Ieoo-def}, as
well as the evaluation in~\eqref{eq:Ide-eval-G2b}, we find
\begin{align}
	e^{2F(\Id^{(\oo)})}/\Id^{(\oo)}  &= \ff''(z) = \cF''(z-z_0) , \\
	\frac{\Ie^{(\oo)}}{\Id^{(\oo)}} e^{F(\Id^{(\oo)})}
		&= -\overline{(z-z_0)} + \ff'(z)
	\notag \\
		&= -\overline{(z-z_0)} + A_{1,0} + \cF'(z-z_0) ,
\end{align}
where $z = \zeta e^{i\lambda u}$. Combining the two formulas, we can
extract the constant
\begin{align} \label{eq:AooF-def}
	c = |A_{1,0}|
	&= \left|\frac{\Ie^{(\oo)}}{\Id^{(\oo)}} e^{F(\Id^{(\oo)})}
		+ \overline{Z_b} - \cF'(Z_b)\right|
	=: \cA_{\oo,\Re F}(\Id^{(\oo)}, \Ie^{(\oo)}) , \\
	\text{where} \quad
	Z_b
	&= (\cF'')^{-1}\left(e^{2F(\Id^{(\oo)})}/\Id^{(\oo)}\right) .
\end{align}
While the individual terms entering formula $\cA_{\oo,\Re F}$ that extracts
$c = |A_{1,0}|$ are not invariant, since they depend on the choice of $F(y)$
for given $\Re F(y)$, it is straightforward to check that the whole
formula is independent of those choices and so defines an invariant.
Still, for given $\Re F(y)$, to make the formula really explicit, we
need to pick a compatible $F(y)$ and determine the unique solution
$\cF(z)$ from it.

The only possibilities left to consider are those with, separately,
$\del_z (\ff'''(z)^2/\ff''(z)^3) = 0$ and $\lambda = 0$.

\item \label{itm:G2b00} 
Within the $f(\zeta,u) = \ff(\zeta e^{i\lambda u})$ formfactor family
under the condition $\del_z(\ff'''(z)^2/(\ff''(z)^3) = 0$, let us
parametrize
\begin{equation} \label{eq:G2b00-entry}
	\frac{\ff'''(z)^2}{\ff''(z)^3} = - 4a^{-2} e^{-i\gamma} ,
\end{equation}
with constant $\alpha > 0$ and $\gamma \in [-\pi,\pi]$, which integrates
to
\begin{equation}
	\ff(z) = a^2 e^{i\gamma} \ln z + A_{1,0} z + A_{0,0} .
\end{equation}
We can set $A_{0,0} = 0$ and $a=1$ by residual gauge freedom, which then
makes the resulting value of $\lambda$ invariant. Further, we can
eliminate the $\ff(z) \mapsto \bar{\ff}(z)$ gauge freedom by choosing
$\gamma \in [0,\pi]$, and also eliminate the $\zeta \mapsto e^{i\beta}
\zeta$ gauge freedom by choosing $A_{1,0} = c \ge 0$. The resulting
formfactor takes the form
\begin{equation}
	f(\zeta,u) = e^{i\gamma} \ln \zeta e^{i\lambda u} + c \zeta e^{i\lambda u} ,
\end{equation}
or after applying another remaining gauge transformation
\begin{equation}
	f(\zeta,u) = e^{i\gamma} \ln \zeta + c \zeta e^{i\lambda} ,
\end{equation}
with $\gamma \in [0,\pi]$, $\lambda \in \mathbb{R}$ and $c \ge 0$
labelling the corresponding isometry class. When $c \ne 0$, we
label this isometry class as $G_{2b'';\gamma,\lambda, c}$.

The excluded value $c = |A_{1,0}| = 0$ lands us in the $G_{3b}$ subclass.

\item \label{itm:G3b} 
Within the formfactor subfamily
\begin{equation}
	f(\zeta,u) = e^{i\gamma} \ln \zeta + c \zeta e^{i\lambda} ,
\end{equation}
where $c = 0$, the $\lambda$ parameter
drops out and the formfactor reduces to
\begin{equation}
	f(\zeta,u) = e^{i\gamma} \ln \zeta .
\end{equation}
We can use the $\zeta \mapsto \bar{\zeta}$ transformation to restrict
$\gamma \in [0,\pi]$, which then becomes invariant. When $e^{i\gamma}
\ne \pm1$, we label this isometry class as $G_{3b^\circ;\gamma}$.

The excluded value of $e^{i\gamma} = \pm1$ lands us in the $G_{3b';\pm}$
subclass.

\item \label{itm:G3b0} 
In the formfactor subfamily
\begin{equation}
	f(\zeta,u) = \pm \ln \zeta ,
\end{equation}
the only invariant is the sign $\pm$. We label this isometry class as
$G_{3b';\pm}$.

\item \label{itm:G2a} 
Consider the invariant equation
\begin{equation} \label{eq:G2a-entry}
	\left((u-u_0) \frac{(\ln f_{,\zeta\zeta})_{,u}}{(\ln f_{,\zeta\zeta})_{,\zeta}}\right)_{,\zeta}
	= \left(\frac{2i\kappa-2}{(\ln f_{,\zeta\zeta})_{,\zeta}}\right)_{,\zeta}
		+ i\kappa ,
\end{equation}
where $\kappa\in \mathbb{R}$ is an invariant parameter. Its first
integral is
\begin{multline}
	(u-u_0) \frac{(\ln f_{,\zeta\zeta})_{,u}}{(\ln f_{,\zeta\zeta})_{,\zeta}}
		- \frac{2i\kappa-2}{(\ln f_{,\zeta\zeta})_{,\zeta}}
	= i\kappa \zeta - u^{-i\kappa + 1} h'(u)
	\\ \iff
	\frac{(u-u_0)^{i\kappa+1} (\ln (u-u_0)^{-2i\kappa + 2} f_{,\zeta\zeta})_{,u}}
		{(\ln (u-u_0)^{-2i\kappa + 2} f_{,\zeta\zeta})_{,\zeta}}
	= (\zeta i\kappa (u-u_0)^{i\kappa} - (u-u_0) h'(u))
	\\  \hspace{-14em} \iff
	\left[(\zeta (u-u_0)^{i\kappa} - h(u))_{,\zeta} (u-u_0)\del_u
	\right. \\ \left. {}
		- (u-u_0) (\zeta (u-u_0)^{i\kappa} - h(u))_{,u} \del_\zeta\right]
		\ln (u-u_0)^{-2i\kappa+2} f_{,\zeta\zeta} = 0 ,
\end{multline}
with implicitly defined integration parameter $h(u)$. The general
solution to the last equation is clearly
\begin{multline}
	\ln (u-u_0)^{-2i\kappa+2} f_{,\zeta\zeta}
		= \ln \ff''\left(\zeta (u-u_0)^{i\kappa} - h(u)\right)
	\\ \iff
	f(\zeta,u) = \frac{1}{(u-u_0)^2}
		\ff\left(\zeta (u-u_0)^{i\kappa} - h(u)\right)
		+ A_1(u) \zeta + A_0(u) ,
\end{multline}
for some holomorphic function $\ff(z)$ and integration parameters
$A_0(u)$, $A_1(u)$. We can set $h(u) = A_0(u) = u_0 = 0$ by gauge
freedom. For this form of $f(\zeta,u)$, the value of the curvature
invariant $\bM^{(i\kappa)}$ in~\eqref{eq:Mk-def} is given
in~\eqref{eq:M-eval-G2a}. So setting $\bM^{(i\kappa)} = 0$ forces
$\ff'''(z) = 0$ or $A_1(u) \zeta = A_{1,0} \zeta u^{i\kappa-2}$. The
former possibility results in a formfactor corresponding to $G_5$
subclass, while in the latter possibility the term $A_1(u)\zeta$ can be
absorbed into the definition of $\ff(z)$, which lands us in the $G_{2a}$
subclass.

In this branch of the classification, we assume now that $\ff'''(z) \ne
0$, and the above conditions hold with $\kappa \in \mathbb{R}$, so that
we are in the $G_{2a}$ subclass.

Comparing~\eqref{eq:G2a-entry} to the definition of the invariant
$\Xi^{(i\kappa)}$ in~\eqref{eq:Xik-def}, we see that it implies that the
invariant $\Xi^{(i\kappa)} = 1/u$. The residual gauge transformations
that preserve the formfactor $f(\zeta,u) = u^{-2} \ff(\zeta
u^{i\kappa})$ are $\zeta \mapsto e^{i\beta}\zeta$, which induces $\ff(z)
\mapsto \ff(e^{-i\beta} z)$, $\zeta \mapsto \zeta + z_0 u^{-i\kappa}$,
which induces $\ff(z) \mapsto \ff(z-z_0) + i\kappa(i\kappa-1)
\overline{z_0} z + \frac{1}{2} |z_0|^2$,
$v \mapsto v + 2u^{-1}\Re w_0$ with $c(u) = u^{-2}\Im w_0$, which
induces $\ff(z) \mapsto \ff(z) + w_0$, and $\zeta \mapsto \bar{\zeta}$,
which induces $\ff(z) \mapsto \bar{\ff}(z)$.

Generically, $\del_z(\ff'''(z)^2/\ff''(z)^3) \ne 0$, so then there
exists a holomorphic function $F\colon U \subset \mathbb{C} \to
\mathbb{C}$ such that
\begin{equation} \label{eq:G2a-signature}
	\ln \frac{\ff'''(z)}{\ff''(z)}
	= F\left(\frac{\ff'''(z)^2}{\ff''(z)^3}\right) .
\end{equation}
The argument of $F(-)$ is an invariant because it coincides with the
value of the invariant $\Id^{(i\kappa)}$ in~\eqref{eq:Idk-def}, which
was evaluated in~\eqref{eq:Ide-eval-G2a}. Neither the left-hand side
in~\eqref{eq:G2a-signature} nor $F(y)$ itself are entirely invariant, as
$\ff(z) \mapsto \ff(e^{-i\beta} z)$ induces $F(y) \mapsto F(y) - i\beta$
and $\ff(z) \mapsto \bar{\ff}(z)$ induces $F(y) \mapsto \bar{F}(y)$.
However, $\Re F(y)$ is invariant under these transformations and also
uniquely fixes $F(y)$ up to these ambiguities. Thus, we consider $\Re
F\colon U \subset \mathbb{C} \to \mathbb{R}$ as our signature function.

If $F'(y) = 1/(2y)$ or equivalently $F(y) = \frac{1}{2}\ln (2A y)$, then
the ODE~\eqref{eq:G2a-signature} no longer depends on $\ff'''(z)$ and
reduces to $\ff''(z) = 2A$, and hence $\ff'''(z) = 0$, which contradicts
our earlier hypothesis. Otherwise, the solution $\ff(z) = A z^2 + A_1 z
+ A_0$ recovers the $G_{6a}$ subclass.

Provided $F'(y) \ne 1/(2y)$, the ODE~\eqref{eq:G2a-signature} can be
solved for $\ff'''(z)$, so the usual existence and uniqueness theory
holds for it. If $\ff(z) = \cF(z)$ is a particular solution, then
symmetries dictate that the general 3-parameter solution is
\begin{equation}
	\ff(z) = \cF(z-z_0) + A_{1,0} (z-z_0) + A_{0,0} .
\end{equation}
While the $z_0$ and $A_{0,0}$ integration constants can be absorbed by
the residual gauge freedom, the parameter $A_{1,0}$ cannot be and does
contain some invariant information.
To be concrete, we uniquely fix%
	\footnote{There may be other ways to uniquely fixing the solution
	$\cF(z)$, for instance by different initial conditions at $z=0$ or by
	asyptotic conditions at $z=\oo$.} %
$\cF(z)$ from $F(y)$ by requiring that $\cF(z) = \frac{1}{2} z^2 +
O(z^3)$. The gauge transformation $F(y) \mapsto F(y) - i\beta$ induces
$\cF(z) \mapsto \cF(e^{-i\beta} z)$, $A_{1,0} \mapsto A_{1,0}
e^{-i\beta}$, and $F(y) \mapsto \bar{F}(y)$ induces $\cF(z) \mapsto
\overline{\cF}(z)$, $A_{1,0} \mapsto \overline{A_{1,0}}$. So only
$c = |A_{1,0}|$ is invariant contained in $A_{1,0}$. Using the
ODE~\eqref{eq:G2a-signature}, and computing the invariants $\Id^{(\oo)}$
and $\Ie^{(\oo)}$, from~\eqref{eq:Idk-def} and~\eqref{eq:Iek-def}, as
well as the evaluation in~\eqref{eq:Ide-eval-G2a}, we find
\begin{align}
	-e^{2F(\Id^{(i\kappa)})}/\Id^{(i\kappa)}  &= \ff''(z) = \cF''(z-z_0) , \\
	-\frac{\Ie^{(i\kappa)}}{\Id^{(i\kappa)}} e^{F(\Id^{(\oo)})}
		&= i\kappa (i\kappa-1)\overline{(z-z_0)} + \ff'(z)
	\notag \\
		&= i\kappa (i\kappa-1)\overline{(z-z_0)} + A_{1,0} + \cF'(z-z_0) ,
\end{align}
where $z = \zeta u^{i\kappa}$. Combining the two formulas, we can
extract the constant
\begin{align} \label{eq:AkF-def}
	c = |A_{1,0}|
	&= \left|\frac{\Ie^{(i\kappa)}}{\Id^{(i\kappa)}} e^{F(\Id^{(i\kappa)})}
		- i\kappa(i\kappa-1) \overline{Z_a} - \cF'(Z_a)\right|
	\notag \\
	&\quad =: \cA_{\kappa,\Re F}(\Id^{(i\kappa)}, \Ie^{(i\kappa)}) , \\
	\text{where} \quad
	Z_a
	&= (\cF'')^{-1}\left(-e^{2F(\Id^{(i\kappa)})}/\Id^{(i\kappa)}\right) .
\end{align}
While the individual terms entering formula $\cA_{\kappa,\Re F}$
that extracts $c = |A_{1,0}|$ are not invariant, since they depend on the
choice of $F(y)$ for given $\Re F(y)$, it is straightforward to check
that the whole formula is independent of those choices and so defines an
invariant. Still, for given $\Re F(y)$, to make the formula really
explicit, we need to pick a compatible $F(y)$ and determine the unique
solution $\cF(z)$ from it.

The only possibility left to consider is that with $\del_z
(\ff'''(z)^2/\ff''(z)^3) = 0$.

Suppose we reparametrize $\Re F(y) = \Re \tilde{F}(\kappa^2 y)$. Then
the solutions of the signature ODE~\eqref{eq:G2a-signature} scale as
$\ff(z) = \kappa^2 \tilde{\ff}(z)$, where $\tilde{\ff}(z)$ solves the
signature ODE for $\Re\tilde{F}(y)$. Taking the invariant parameter
limit $\kappa, c\to \oo$, accompanied by the gauge parameter
limits $u_0,\beta \to -\oo$, such that $(-u_0) e^{\beta/\kappa} \to 1$,
$\kappa/(-u_0) \to 1$ and $c/\kappa^2 \to \tilde{c}$, we
obtain
\begin{equation}
	\zeta e^{i\beta} (u-u_0)^{i\kappa}
		\to \zeta e^{iu}
	\quad \text{and} \quad
	\frac{1}{(u-u_0)^2} \ff(z)
		\to \tilde{\ff}(z) ,
\end{equation}
and hence we recover the $G_{2b^\circ;\tilde{c},\Re \tilde{F}}$
subclass.

\item \label{itm:G2a0} 
Within the $f(\zeta,u) = u^{-2} \ff(\zeta u^{i\kappa})$ formfactor
family under the condition $\del_z(\ff'''(z)^2/\ff''(z)^3) = 0$, let us
parametrize
\begin{equation} \label{itm:G2a0-entry}
	\frac{\ff'''(z)^2}{\ff''(z)^3} = -\alpha e^{-i\gamma} ,
\end{equation}
with constant $\alpha > 0$ and $\gamma \in [-\pi,\pi]$, which integrates
to
\begin{equation}
	\ff(z) = 4\alpha e^{i\gamma} \ln z + A_{1,0} z + A_{0,0} .
\end{equation}
We can set $A_{0,0} = 0$ by residual gauge freedom. Further, we can
eliminate the $\ff(z) \mapsto \bar{\ff}(z)$ gauge freedom by choosing
$\gamma \in [0,\pi]$, and also eliminate the $\zeta \mapsto e^{i\beta}
\zeta$ gauge freedom by choosing $A_{1,0} = c \ge 0$. The resulting
formfactor takes the form
\begin{equation}
	f(\zeta,u) = \frac{4\alpha}{u^2} e^{i\gamma} \ln \zeta e^{i\lambda u}
		+ c \zeta u^{i\kappa} ,
\end{equation}
or after applying another remaining gauge transformation
\begin{equation}
	f(\zeta,u) = \frac{4\alpha}{u^2} \left(e^{i\gamma} \ln \zeta
		+ c \zeta u^{i\kappa} \right) ,
\end{equation}
with $\alpha > 0$, $\gamma \in [0,\pi]$, $\kappa \in \mathbb{R}$ and
$c \ge 0$ labelling the corresponding isometry class. When
$c \ne 0$, we label this isometry class as
$G_{2a';\alpha,\gamma,\kappa, c}$.

The excluded value of $c = |A_{1,0}| = 0$ lands us in the $G_{3a}$ subclass.

Taking the invariant parameter limit $\kappa, \alpha \to \oo$,
accompanied by the gauge parameter limits $u_0,\beta \to -\oo$, such
that $4\alpha/u_0^2 \to 1$, $(-u_0) e^{\beta/\kappa} \to 1$ and
$\kappa/(-u_0) \to \lambda$, we recover the
$G_{2b'';\gamma,\lambda,c}$ subclass.

\item \label{itm:G3a} 
Within the formfactor subfamily
\begin{equation}
	f(\zeta,u) = \frac{4\alpha}{u^2} \left(e^{i\gamma} \ln \zeta + c \zeta u^{i\kappa}\right),
\end{equation}
where $c = 0$, the $\kappa$ parameter drops out and the formfactor
reduces to
\begin{equation}
	f(\zeta,u) = \frac{4\alpha}{u^2} e^{i\gamma} \ln \zeta .
\end{equation}
We can use the $\zeta \mapsto \bar{\zeta}$ transformation to restrict
$\gamma \in [0,\pi]$, which then becomes invariant, along with $\alpha$.
When $e^{i\gamma} \ne \pm1$, we label this isometry class as
$G_{3a^\circ;\gamma}$.

The excluded value of $e^{i\gamma} = \pm1$ lands us in the
$G_{3a';\pm\alpha}$ subclass.

Taking the invariant parameter limit $\alpha\to \oo$, accompanied by the
gauge parameter limit $u_0 \to \oo$, such that $4\alpha/u_0^2 \to 1$, we
recover the $G_{3b^\circ;\gamma}$ subclass.

\item \label{itm:G3a0} 
In the formfactor subfamily
\begin{equation}
	f(\zeta,u) = \pm \frac{4\alpha}{u^2} \ln \zeta ,
\end{equation}
the only invariants are the sign $\pm$ and $\alpha>0$. We label this
isometry class as $G_{3a';\pm\alpha}$.

Taking the invariant parameter limit $\alpha\to \oo$, accompanied by the
gauge parameter limit $u_0 \to \oo$, such that $4\alpha/u_0^2 \to 1$, we
recover the $G_{3b';\pm}$ subclass.

\item \label{itm:G2b0} 
Suppose that within the $G_{2b}$ subclass, the following invariant
condition holds
\begin{equation} \label{eq:G2b0-entry}
	\left(\frac{(\ln f_{,\zeta\zeta})_{,u}}{(\ln f_{,\zeta\zeta})_{,\zeta}}\right)_{\zeta}
	= 0 .
\end{equation}
As discussed earlier, it corresponds to either $-(1/(\ln
f_{,\zeta\zeta})_{,\zeta}) = \frac{1}{2}$, which lands us in the
$G_{2b''}$ subclass, or $\lambda=0$. In this branch of the
classification we consider the latter option, which corresponds to the
formfactor family
\begin{equation}
	f(\zeta,u) = \ff(\zeta) .
\end{equation}
The evaluation~\eqref{eq:Ide-eval-G2b} of the invariant~$\Id^{(\oo)}$
with $\lambda=0$ simplifies to
\begin{equation}
	\Ie^{(\oo)} = \frac{\ff'''(z) \ff'(z)}{\ff''(z)^2} ,
\end{equation}
where $z = \zeta$. Generically $\del_z (\ff'''(z) \ff'(z)/\ff''(z)^2) =
0$, so then there exists a holomorphic function $F\colon U\subset
\mathbb{C} \to \mathbb{C}$ such that
\begin{equation} \label{eq:G2b0-signature}
	\ln \frac{\ff'''(z)}{\ff''(z)}
		= F\left(\frac{\ff'''(z) \ff'(z)}{\ff''(z)^2}\right) .
\end{equation}
Neither the left-hand side above nor $F(y)$ are entirely invariant, as
$\ff(z) \mapsto \ff(e^{-i\beta} z)$ induces $F(y) \mapsto F(y) - i\beta$
and $\ff(z) \mapsto \bar{\ff}(z)$ induces $F(y) \mapsto \bar{F}(y)$.
However, $\Re F(y)$ is invariant under these transformations and also
uniquely fixes $F(y)$ up to these ambiguities. Thus, we consider $\Re
F\colon U \subset \mathbb{C} \to \mathbb{R}$ as our signature function.

If $F'(y) = 1/y$ or equivalently $F(y) = \ln(\lambda y)$, then the
ODE~\eqref{eq:G2b0-signature} no longer depends on $\ff'''(z)$ and
reduces to $\ff''(z) = \lambda \ff'(z)$, and hence $\ff'''(z)
\ff'(z)/\ff''(z)^2 = 1$, which contradicts our genericity hypothesis. If
$F'(y) = 1/(2-y)$, or equivalently $F(y) = \ln\frac{4c}{2-y}$,
then the ODE~\eqref{eq:G2b0-signature} has the general 3-parameter
solution
\begin{equation}
	\ff(z) = A(\ln(z-z_0) + c z + A_{0,0}) ,
\end{equation}
which lands us in the $G_{2b''}$ subclass, as well as a singular
2-parameter solution
\begin{equation}
	\ff(z) = e^{4c(z-z_0)} + A_{0,0} ,
\end{equation}
which also implies $\ff'''(z) \ff'(z)/\ff''(z)^2 = 1$ and hence
contradicts the genericity hypothesis.

Provided $F'(y) \ne 1/y$, the ODE~\eqref{eq:G2b0-signature} can be solved for $\ff''(z)$, so
the usual existence and uniqueness theory holds for it. If $\ff(z) =
\cF(z)$ is a particular solution, then symmetries dictate that the
general 3-parameter solution is
\begin{equation}
	\ff(z) = A \cF(z-z_0) + A_{0,0} .
\end{equation}
Writing $A = a^2 e^{i\gamma}$ with $a>0$ and $\gamma\in[-\pi,\pi]$, the
$a$ and $A_{0,0}$ integration constants can be absorbed by the residual
gauge freedom. To be concrete, we uniquely fix%
	\footnote{There may be other ways to uniquely fixing the solution
	$\cF(z)$, for instance by different initial conditions at $z=0$ or by
	asyptotic conditions at $z=\oo$.} %
$\cF(z)$ from $F(y)$ by requiring that $\cF(z) = \frac{1}{2} z^2 +
O(z^3)$. The gauge transformation $F(y) \mapsto F(y) - i\beta$ induces
$\cF(z) \mapsto \cF(e^{-i\beta} z)$, and $F(y) \mapsto \bar{F}(y)$
induces $\cF(z) \mapsto \overline{\cF}(z)$, $\gamma \mapsto -\gamma$. So
only $|\gamma| \in [0,\pi]$ is invariant among the integration
constants. Using the signature ODE~\eqref{eq:G2b0-signature}, and
computing the invariants $J$ and $\Ie^{(\oo)}$, from~\eqref{eq:J-def}
and~\eqref{eq:Ieoo-def}, we find
\begin{align}
	\Ie^{(\oo)} e^{-F(\Ie^{(\oo)})}
		&= \frac{\ff'(z)}{\ff''(z)}
		= \frac{\cF'(z-z_0)}{\cF''(z-z_0)} ,
	\\
	J e^{4\Im F(\Ie^{(\oo)})}
		&= \frac{\ff''(z)}{\overline{\ff''(z)}}
		= e^{2i\gamma} \frac{\cF''(z-z_0)}{\overline{\cF''(z-z_0)}} .
\end{align}
Combining the two formulas, we can extract the constant
\begin{align} \label{eq:GF-def}
	\cos\gamma
	&= \Re \left[J e^{4\Im F(\Ie^{(\oo)})}
		\frac{\overline{\cF''(Z_{b'})}}{\cF''(Z_{b'})}\right]^{1/2}
	=: \cG_{\Re F}(\Ie^{(\oo)}) ,
	\\ \text{where} \quad
	Z_{b'} &= \left[\frac{\cF'(-)}{\cF''(-)}\right]^{-1}
			\left(\Ie^{(\oo)} e^{-F(\Ie^{(\oo)})}\right) .
\end{align}
Note that the value of $\cos\gamma$ fixes a unique $\gamma \in [0,\pi]$,
which we combine with the square root prescription $[e^{2i\gamma}]^{1/2}
= e^{i\gamma}$ with $\gamma \in [0,\pi]$. While the individual terms
entering formula $\cG_{\Re F}$ that extracts $\cos\gamma$ are
not invariant, since they depend on the choice of $F(y)$ for given $\Re
F(y)$, it is straightforward to check that the whole formula is
independent of those choices and so defines an invariant. Still, for
given $\Re F(y)$, to make the formula really explicit, we need to pick a
compatible $F(y)$ and determine the unique solution $\cF(z)$ from it.

The only possibilities left to consider are those with constant
$\ff'''(z) \ff'(z)/\ff''(z)^2 = \frac{k-2}{k-1}$, separately for $k=\oo$
and $k\ne \oo$.

\item \label{itm:G3c} 
Within the $f(\zeta,u) = \ff(\zeta)$ formfactor family, consider the
invariant condition
\begin{equation}
	\Ie^{(\oo)} = \frac{\ff'''(z) \ff'(z)}{\ff''(z)^2} = 1 ,
\end{equation}
where we have replaced $z = \zeta$, which integrates to the general
$3$-parameter solution
\begin{equation}
	\ff(z) = e^{2\lambda(z-z_0)} + A_{0,0}.
\end{equation}
We can set the integration constants $z_0 = A_{0,0} = 0$ by residual
gauge freedom, as well as use it to select $\lambda > 0$, whose value
then becomes the only invariant of this isometry class, which we label
as $G_{3c;\lambda}$.

\item \label{itm:G2b-zk} 
Within the $f(\zeta,u) = \ff(\zeta)$ formfactor family, consider the
invariant condition (with $k\ne \oo$)
\begin{equation}
	\Ie^{(\oo)} = \frac{\ff'''(z) \ff'(z)}{\ff''(z)^2} = \frac{k-2}{k-1} ,
\end{equation}
where we have replaced $z = \zeta$, which integrates to the general
$3$-parameter solution
\begin{equation}
	\ff(z) = \begin{cases}
		a^2 e^{i\gamma} (z-z_0)^{k} + A_{0,0} & (k\ne 0,1) \\
		a^2 e^{i\gamma} \ln(z-z_0) + A_{0,0} & (k=0)
	\end{cases} .
\end{equation}
When $k=0$, we recover the $G_{2b'''}$ subclass,
while the value $k = 1$ is not allowed by the previous formula. When
$k=2$ we recover the $G_{6b}$ subclass. When $k\not\in i\mathbb{R} \cup
\{1,2\}$, we can set $a^2 e^{i\gamma} = 1$ and $z_0 = A_{0,0} = 0$ by
gauge freedom. We label the resulting isometry class by $G_{2b''';k}$.

The only possibility left to consider is that with $k \in i\mathbb{R}$.

\item \label{itm:G3d} 
Within the formfactor family
\begin{equation}
	f(\zeta,u) = a^2 e^{i\gamma} \zeta^{k} ,
\end{equation}
with $k=2i\kappa$, $\kappa \in \mathbb{R} \setminus \{0\}$, we can set
$a^2=1$ by residual gauge freedom, while by exploiting the $\zeta
\mapsto \bar{\zeta}$ gauge transformation, we can restrict $\gamma \in
[0,\pi]$, which then becomes an invariant. We label the corresponding
isometry class as $G_{3d;\gamma,\kappa}$.

Taking the invariant parameter limit $\gamma,\kappa \to \oo$,
accompanied by the gauge parameter limit $a, z_0 \to \oo$, such that
$\gamma - 2i\ln a + 2\kappa\ln (-z_0) \to 0$ and $\kappa/(-z_0) =
-i\lambda$, we obtain
\begin{equation}
	e^{i\gamma} \zeta^{2i\kappa} \to e^{2\lambda\zeta}
\end{equation}
and hence we recover the $G_{3c;\lambda}$ subclass.

\end{enumerate}
Together with the data in Figure~\ref{fig:pp-wave-classes} and
Table~\ref{tab:pp-wave-isom-classes} this completes the classification
of local isometry classes of regular highly symmetric pp-wave
geometries.
\end{proof}

\setlength{\LTcapwidth}{\textwidth}%
\begin{longtable}{lcc}
\caption{IDEAL characterization of highly symmetric
	pp-wave isometry classes to accompany Figure~\ref{fig:pp-wave-flowchart} and
	Theorem~\ref{thm:ideal-pp-wave-isom}. The details of the notation and
	the decision logic are contained in the statement and proof of the
	theorem.}
\label{tab:pp-wave-ideal}\\
	class & parameters & inequalities/equations \\
	\hline\hline
	\endfirsthead
	\caption[]{(continued)}\\
	class & parameters & inequalities/equations \\
	\hline\hline
	\endhead
	\noalign{\smallskip}
	$G_{6a;\alpha,\kappa}$: &
		$\alpha>0$, $\kappa\ge 0$ &
		$|\bK|^2 = 0$,
		$\Ia^{(2)} = (i\kappa - 1) \alpha^{-1/2}$
	\\ \noalign{\smallskip}\hline\noalign{\smallskip}
	$G_{6b;\lambda}$: &
		$\lambda\ge 0$ &
		$\begin{gathered}
			|\bK|^2 = 0 , \quad
			\Ia^{(2)} = 2i\lambda
		\end{gathered}$
	\\ \noalign{\smallskip}\hline\noalign{\smallskip}
	$G_{5';\alpha,F}$: &
		$\oo\ge\alpha>0$, $F(y)$ &
		$\nabla\Im \Ia^{(2)} \ne 0$
			\\ \noalign{\smallskip}\cline{3-3}\noalign{\smallskip}
		& &
		$\begin{gathered}
			|\bK|^2 = 0 , \quad
			\Re \Ia^{(2)} = -\alpha^{-1/2} , \\
			\Im \Ib^{(2)} = F(\Im \Ia^{(2)})
		\end{gathered}$
	\\ \noalign{\smallskip}\hline\noalign{\smallskip}
	$G_{5^\circ;F}$: &
		$F(y)$ &
		$\nabla\Re \Ia^{(2)} \ne 0$
			\\ \noalign{\smallskip}\cline{3-3}\noalign{\smallskip}
		& &
		$\begin{gathered}
			|\bK|^2 = 0, \\
			\begin{pmatrix}
				\Re \Ib^{(2)} \\ (\Im \Ia^{(2)})^2
			\end{pmatrix} = F(\Re \Ia^{(2)})
		\end{gathered}$
	\\ \noalign{\smallskip}\hline\noalign{\smallskip}
	$G_{3a';\pm\alpha}$: &
		$\alpha>0$, $\{\pm\}$ &
		$|\bK|^2 \ne 0$
			\\ \noalign{\smallskip}\cline{3-3}\noalign{\smallskip}
		& &
		$\begin{gathered}
			\Ic = 1/2 , \quad
			\bN  = 0 , \\
			\Re\Ia^{(0)} = -\alpha^{-1/2} , \\
			J = 1 , \quad
			\Id^{(0)} = \pm\alpha^{-1/2}
		\end{gathered}$
	\\ \noalign{\smallskip}\hline\noalign{\smallskip}
	$G_{3a^\circ;\alpha,\gamma}$: &
		$\alpha>0$, $\gamma \in (0,\pi)$ &
		$|\bK|^2 \ne 0$
			\\ \noalign{\smallskip}\cline{3-3}\noalign{\smallskip}
		& &
		$\begin{gathered}
			\Ic = 1/2 , \quad
			\bN  = 0 , \\
			\Re\Ia^{(0)} = -\alpha^{-1/2} , \quad
			J = e^{2i\gamma}
		\end{gathered}$
	\\ \noalign{\smallskip}\hline\noalign{\smallskip}
	$G_{3b';\pm}$: &
		$\{\pm\}$ &
		$|\bK|^2 \ne 0$
			\\ \noalign{\smallskip}\cline{3-3}\noalign{\smallskip}
		& &
		$\begin{gathered}
			\Ic = 1/2 , \quad
			\bN  = 0 , \\
			\Re\Ia^{(0)} = 0 , \quad
			J = 1 , \\
			\exists \bell: -2\bD = \pm \bell \bell
		\end{gathered}$
	\\ \noalign{\smallskip}\hline\noalign{\smallskip}
	$G_{3b^\circ;\gamma}$: &
		$\gamma\in (0,\pi)$ &
		$|\bK|^2 \ne 0$
			\\ \noalign{\smallskip}\cline{3-3}\noalign{\smallskip}
		& &
		$\begin{gathered}
			\Ic = 1/2 , \quad
			\bN  = 0 , \\
			\Re\Ia^{(0)} = 0 , \quad
			J = e^{2i\gamma}
		\end{gathered}$
	\\ \noalign{\smallskip}\hline\noalign{\smallskip}
	$G_{3c;\lambda}$: &
		$\lambda>0$ &
		$|\bK|^2 \ne 0$, $\Ic \ne 1/2$
			\\ \noalign{\smallskip}\cline{3-3}\noalign{\smallskip}
		& &
		$\begin{gathered}
			\fL = 0 , \quad
			\bM^{(\oo)} = 0 , \\
			\Ie^{(\oo)} = 1 , \quad
			|\bK|^2 = 4\lambda^2
		\end{gathered}$
	\\ \noalign{\smallskip}\hline\noalign{\smallskip}
	$G_{3d;\gamma;\kappa}$: &
		$\gamma \in [0,\pi]$, $\kappa\ge 0$ &
		$|\bK|^2 \ne 0$, $\Ic \ne 1/2$
			\\ \noalign{\smallskip}\cline{3-3}\noalign{\smallskip}
		& &
		\footnotesize$\begin{gathered}
			\fL = 0 , \quad
			\bM^{(\oo)} = 0 , \quad
			\Ie^{(\oo)} = \frac{2(i\kappa - 1)}{2i\kappa-1} , \\
			J = e^{2i\gamma} \frac{(2i\kappa-1) (i\kappa+1)^2}{(2i\kappa+1) (i\kappa-1)^2}
				\left|\frac{\bK}{2(i\kappa-1)}\right|^{-4i\kappa}
		\end{gathered}$
	\\ \noalign{\smallskip}\hline\noalign{\smallskip}
	$G_{2b''';k}$: &
		\multirow[t]{2}{*}{\parbox[t]{5em}{\centering
			$k \in \mathbb{C}$,\\
			$\Re k \ne 0$, $k\ne 1,2$}}
		&
		$|\bK|^2 \ne 0$, $\Ic \ne 1/2$
			\\ \noalign{\smallskip}\cline{3-3}\noalign{\smallskip}
		& &
		$\begin{gathered}
			\fL = 0 , \quad
			\bM^{(\oo)} = 0 , \quad
			\Ie^{(\oo)} = \frac{k-2}{k-1}
		\end{gathered}$
	\\ \noalign{\smallskip}\hline\noalign{\smallskip}
	$G_{2b'';\gamma,\lambda,c}$: &
		\multirow[t]{2}{*}{\parbox[t]{7.3em}{\centering
			$\lambda\in \mathbb{R}$, $\gamma\in [0,\pi]$,\\[2pt]
			$c > 0$}}
		&
		$|\bK|^2 \ne 0$
			\\ \noalign{\smallskip}\cline{3-3}\noalign{\smallskip}
		& &
		\small$\begin{gathered}
			\Id = 1/2 , \quad
			\fL = 0, \quad
			|N| = c |\bK| \\
			\left(\nabla N - \frac{N}{2} \overline{\bK}\right)^{\otimes 2}
				= 2\lambda e^{-i\gamma \sgn\lambda} N^2 \bD^{\dagger} , \\
			\begin{cases}
				\bD = 2e^{i\gamma} \overline{\bD}
					& \text{if $e^{i\gamma} \ne \pm 1$} \\
				-2\bD = \pm\bell\bell
					& \text{if $e^{i\gamma} = \pm 1$}
			\end{cases}
		\end{gathered}$
	\\ \noalign{\smallskip}\hline\noalign{\smallskip}
	$G_{2a';\alpha,\gamma,\kappa,c}$: &
		\multirow[t]{2}{*}{\parbox[t]{7.3em}{\centering
			$\alpha>0$, $\gamma\in [0,\pi]$,\\[2pt]
			$c > 0$}}
		&
		$|\bK|^2 \ne 0$, $\fL\ne 0$
			\\ \noalign{\smallskip}\cline{3-3}\noalign{\smallskip}
		& &
		$\begin{gathered}
			\Id = 1/2 , \quad
			\bM^{(i\kappa)} = 0 , \\
			|N| = c |\bK|
		\end{gathered}$
	\\ \noalign{\smallskip}\hline\noalign{\smallskip}
	$G_{2a^\circ;\kappa,c,\Re F}$: &
		\multirow[t]{2}{*}{\parbox[t]{7.3em}{\centering
			$\kappa>0$, $c>0$,\\[2pt]
			$\Re F(z)$}}
		&
		$|\bK|^2 \ne 0$,
		$\Ic\ne 1/2$,
		$\fL \ne 0$,
		\\ &&
		$\nabla\frac{\Ic}{(\Ic-\frac{1}{2})} \ne 0$,
		$\nabla \Id^{(i\kappa)} \ne 0$
			\\ \noalign{\smallskip}\cline{3-3}\noalign{\smallskip}
		& &
		$\begin{gathered}
			2\nabla \bL^{(i\kappa)} = \bL^{(i\kappa)} \bL^{(i\kappa)} , \quad
			\bM^{(i\kappa)} = 0 , \\
			\ln(|\bK|^2) = 2\Re F(-\Id^{(i\kappa)}) , \\
			\cA_{\kappa,\Re F}(\Id^{(i\kappa)}, \Ie^{(i\kappa)}) = c
		\end{gathered}$
	\\ \noalign{\smallskip}\hline\noalign{\smallskip}
	$G_{2b^\circ;c,\Re F}$: &
		$c>0$, $\Re F(z)$ &
		$|\bK|^2 \ne 0$,
		$\Ic\ne 1/2$,
		$\fL \ne 0$,
		\\ &&
		$\nabla\frac{\Ic}{(\Ic-\frac{1}{2})} \ne 0$,
		$\nabla \Id^{(\oo)} \ne 0$
			\\ \noalign{\smallskip}\cline{3-3}\noalign{\smallskip}
		& &
		$\begin{gathered}
			\Re \bL^{(\oo)} = 0 , \quad
			\nabla \bL^{(\oo)} = 0 , \\
			\bM^{(\oo)} = 0 , \\
			\ln(|\bK|^2) = 2\Re F(\Id^{(\oo)}) , \\
			\cA_{\oo,\Re F}(\Id^{(\oo)}, \Ie^{(\oo)}) = c
		\end{gathered}$
	\\ \noalign{\smallskip}\hline\noalign{\smallskip}
	$G_{2b';\gamma,\Re F}$: &
		$\alpha\in [0,\pi]$, $\Re F(z)$ &
		$|\bK|^2 \ne 0$,
		$\Ic\ne 1/2$,
		$\nabla \Ie^{(\oo)} \ne 0$
			\\ \noalign{\smallskip}\cline{3-3}\noalign{\smallskip}
		& &
		$\begin{gathered}
			\fL = 0 , \quad
			\bM^{(\oo)} = 0 , \\
			\ln(|\bK|^2) = 2\Re F(\Ie^{(\oo)}) , \\
			\cG_{\Re F}(\Ie^{(\oo)}) = \cos\gamma
		\end{gathered}$
	\\ \noalign{\smallskip}\hline\noalign{\smallskip}
	$G_{2c';\alpha,F}$: &
		$\oo\ge\alpha>0$, $F(y)$ &
		$|\bK|^2 \ne 0$,
		$\nabla \Im [(\Ia^{(0)})^2/\Id^{(0)}] \ne 0$
			\\ \noalign{\smallskip}\cline{3-3}\noalign{\smallskip}
		& &
		$\begin{gathered}
			\Re \Ia^{(0)} = -\alpha^{-1/2} , \\
			\Im \Ia^{(0)} = F\left(\Im [(\Ia^{(0)})^2/\Id^{(0)}]\right)
		\end{gathered}$
	\\ \noalign{\smallskip}\hline\noalign{\smallskip}
	$G_{2c^\circ;F}$: &
		$F(y)$ &
		$|\bK|^2 \ne 0$,
		$\nabla \Re \Ia^{(0)} \ne 0$
			\\ \noalign{\smallskip}\cline{3-3}\noalign{\smallskip}
		& &
		$\begin{gathered}
			\begin{pmatrix}
				\Re \Ib^{(0)} \\
				\Re [(\Ia^{(0)})^2/\Id^{(0)}]
			\end{pmatrix} = F(\Re\Ia^{(0)})
		\end{gathered}$
	\\ \noalign{\smallskip}\hline\hline
\end{longtable}

\begin{figure}
\centering
\rotatebox[origin=c]{90}{%
\begin{minipage}{.99\textheight} \centering

\tikzstyle{start}=[fill=white, draw=black, shape=trapezium, trapezium left angle=75, trapezium right angle=-75, very thick]
\tikzstyle{decide}=[fill=white, draw=black, shape=chamfered rectangle]
\tikzstyle{leaf}=[fill=white, draw=black, shape=rectangle, rounded corners, thick]
\tikzstyle{info}=[fill=white, draw=black, shape=rectangle]
\tikzstyle{yn}=[fill=white, shape=circle, inner sep=2pt]
\tikzstyle{num}=[fill={{gray!20}}, draw=black, shape=circle, inner sep=1pt, outer sep=3pt]

\tikzstyle{flow}=[->, >={Stealth[length=1.5ex]}, thick]
\tikzstyle{impl}=[->, >=stealth, double]
\tikzstyle{dd}=[-, dashed]
\tikzstyle{line}=[-, thick]
	\resizebox{.99\textheight}{!}{\input{pp-chart-sym.tikz}}
	\caption{Flowchart for the IDEAL characterization of vacuum pp-waves
	of high symmetry, accompanied by Table~\ref{tab:pp-wave-ideal} and
	Theorem~\ref{thm:ideal-pp-wave-isom}.}
	\label{fig:pp-wave-flowchart}
\end{minipage}
}
\end{figure}

\begin{thm}[IDEAL characterization of pp-waves] \label{thm:ideal-pp-wave-isom}
Consider a point $x\in M$, where $(M,g)$ is a Lorentzian metric. Then,
locally at $x$, $(M,g)$ belongs to a given one of the highly symmetric
regular pp-wave isometry types from
Theorem~\ref{thm:pp-wave-isom-classes} if and only if, on a
neighborhood of $x$, the corresponding tensor identities from
Table~\ref{tab:pp-wave-ideal} hold. These identities can isolate every
individual isometry class with one small exception (they cannot
distinguish between $G_{3d;0,\kappa}$ and $G_{3d;\pi,\kappa}$). The
decision process for recognizing and distinguishing these isometry
classes is given by the flowchart in Figure~\ref{fig:pp-wave-flowchart}.
\end{thm}

\begin{proof}
We will proceed by going through the flowchart in
Figure~\ref{fig:pp-wave-flowchart} step by step, in the order of the
numbers indicated next to each node. In the corresponding numbered part
of the proof, for decision nodes, we decode the significance of the
given tensorial condition by mapping it to the corresponding branch
point in the proof of Theorem~\ref{thm:pp-wave-isom-classes} (we
denote such references by IC.\#), and also argue why the arrows leaving
the node exhaust all the possibilities. In some cases, the latter
arguments will rely on the fact that some possibilities may have already
been eliminated by previous branches of the flowchart. For each leaf node,
we decode the tensor identities that extract the invariant parameters of
the isometry class, which are often attached to the branch by a dotted
line, again by mapping them to the relevant point in the proof of
Theorem~\ref{thm:pp-wave-isom-classes}. One can jump directly into
any numbered point in the proof below, but in general to get the full
context one may need to retrace to its origin the corresponding branch
in Figure~\ref{fig:pp-wave-flowchart}.

The entry point is the set of tensorial curvature conditions from
Proposition~\ref{prop:ideal-pp-wave} that identify the vacuum pp-wave
class. For the rest of the diagram, we will assume the corresponding
line element~\eqref{eq:pp-wave-ds2} and will rely on the conditional
curvature invariants defined in Section~\ref{sec:curv-inv}.
Some of the scalar invariants defined there required a choice of branch
of a square or quartic root, which was indicated in the definition. We
point out when these choices affect the possible invariant relations
that can be encountered in the proof, for instance when equating an
invariant to a constant. We do not need to be so careful for functional
relations between invariants, as different branch choices can at worst
be absorbed into the free functional parameters of the relation itself.

\begin{enumerate}
\item 
Since we are already in vacuum, $R_{ab} = 0$, the vanishing $\bC^\dagger
= 0$ of the self-dual Weyl tensor~\eqref{eq:sd-weyl} implies that the
entire Riemann tensor $R_{abcd} = 0$, which lands us in the trivial case
of Minkowski space. It remains to examine the possibilities where
$\bC^\dagger \ne 0$, which by the evaluation~\eqref{eq:bC-expr} is
equivalent to $f_{,\zeta\zeta} \ne 0$.

\item 
The condition $\overline{\bK} \cdot \bK = 0$ by~\eqref{eq:K2-def} is
equivalent to $(\ln f_{,\zeta\zeta})_{,\zeta} =
f_{,\zeta\zeta\zeta}/f_{,\zeta\zeta} = 0$. As we saw in
IC.\ref{itm:K2=0}, this condition lands us in the $G_5$ subclass,
$f(\zeta,u) = A(u)\zeta^2$. It remains to examine both possibilities,
that the equality holds or not.

\item \label{itm:z2-branch} 
The condition $\Re \Ia^{(2)} = -\alpha^{-1/2}$, for some $\alpha > 0$%
	\footnote{Recall that according to the choice of branch
	in~\eqref{eq:Ia2-def} $\Re \Ia^{(2)} \le 0$, so $\alpha$ cannot be
	negative.}, %
results in the three possibilities, $\Re \Ia^{(2)} \ne \text{const}$,
$\alpha = \oo$ and $\alpha \ne \oo$. In the first case, the relation
\begin{equation}
	\begin{pmatrix}
		\Re\Ib^{(2)} \\ (\Im\Ia^{(2)})^2
	\end{pmatrix}
	= F\left(\Re\Ia^{(2)}\right) .
\end{equation}
combined with the evaluation~\eqref{eq:Iab2-eval} coincides with the
signature equation~\eqref{eq:G5-signature} and hence by the arguments
from~IC.\ref{itm:K2=0} it isolates the isometry class $G_{5^\circ;F}$.

It remains to examine the possibilities where $\Re\Ia^{(2)}$ is
constant, including both $\alpha=\oo$ and $\alpha \ne \oo$.

\item 
The condition $\Im\Ia^{(2)} = 2\lambda$ is either satisfied by some
constant $\lambda \ge 0$%
	\footnote{Recall that according to the choice of branch
	in~\eqref{eq:Ia2-def} $\Im \Ia^{(2)} \ge 0$, so $\lambda$ cannot be
	negative.}, %
or $\Im\Ia^{(2)} \ne \text{const}$. In the latter case, the relation
\begin{equation}
	\Im \Ib^{(2)}
	= F\left(\Im\Ia^{(2)}\right) .
\end{equation}
combined with the evaluation~\eqref{eq:Iab2-eval} coincides with the
signature equation~\eqref{eq:G5-ReIa=0-signature} and hence by the
arguments from~IC.\ref{itm:G5-ReIa=0} or~IC.\ref{itm:G5-ReIa=const} it
isolates the isometry class $G_{5';\oo,F}$ or $G_{5';\alpha,F}$,
depending on the value of $\alpha$.

It remains to examine the possibilities $\Ia^{(2)} = 2i\lambda -
\alpha^{-1/2}$. Taking the evaluation~\eqref{eq:Iab2-eval} into account,
the case $\alpha = \oo$ lands us in the $G_{6b;\lambda}$ class
(cf.~IC.\ref{itm:G6b}), while the case $\alpha \ne \oo$ lands us in the
$G_{6a;\alpha,\kappa}$ subclass with $\kappa = 2\lambda \alpha^{1/2}$
(cf.~IC.\ref{itm:G6a}).

\item 
The condition $\Ic = \frac{1}{2}$ combined with the evaluation
in~\eqref{eq:Ic-def} coincides with~\eqref{eq:log-condition}. Hence, by
the arguments from~IC.\ref{itm:G2c}, it lands us in the $f(\zeta,u) =
A(u) \ln\zeta + A_1(u) \zeta + A_0(u)$ subclass. It remains to impose
further conditions to characterize the invariant content of the
$u$-dependent coefficients, or to examine the possibility that $\Ic \ne
\frac{1}{2}$.

\item \label{itm:log-branch} 
The condition $\bN = 0$ combined with the evaluation
in~\eqref{eq:N-eval-G2c} coincides with setting $A_1(u)= 0$ as
in~\eqref{eq:G2c-A1=0}, which by the arguments of~IC.\ref{itm:G2c} lands
us in the $G_{2c}$ subclass, which can be further subclassified based on
the invariant content in the $A(u)$ coefficient. It remains to examine
the possibility that $\bN \ne 0$.

\item 
The condition $\Re \Ia^{(0)} = -\alpha^{-1/2}$, for some $\alpha > 0$%
	\footnote{Recall that according to the choice of branch
	in~\eqref{eq:Ia0-def} $\Re \Ia^{(0)} \le 0$, so $\alpha$ cannot be
	negative.}, %
results in the three possibilities $\Re \Ia^{(0)} \ne
\text{const}$, $\alpha = \oo$ and $\alpha \ne \oo$. In the first case,
the relation
\begin{equation}
	\begin{pmatrix}
		\Re \Ib^{(0)} \\ \Re\left[(\Ia^{(0)})^2/\Id^{(0)}\right]
	\end{pmatrix}
	= F\left(\Re \Ia^{(0)}\right)
\end{equation}
combined with the evaluation~\eqref{eq:Iab0-eval} coincides with the
signature equation~\eqref{eq:G2c-signature} and hence by the arguments
from~IC.\ref{itm:G2c} it isolates the isometry class $G_{2c^\circ;F}$.

\item 
By virtue of its definition~\eqref{eq:J-def}, the invariant $|J| = 1$.
The condition $J = e^{2i\gamma}$ is either satisfied for some unique
constant $\gamma \in [0,\pi)$ or $J\ne \text{const}$. In the latter
case, the relation
\begin{equation}
	\Im \Ia^{(0)}
	= F\left(\Im\left[(\Ia^{(0)})^2/\Id^{(0)}\right]\right) .
\end{equation}
combined with the evaluation~\eqref{eq:Iab0-eval} coincides with the
signature equation~\eqref{eq:G2c0-ReIa=0-signature} and hence by the
arguments from~IC.\ref{itm:G2c0-ReIa=0} and~IC.\ref{itm:G2c0-ReIa=const}
it isolates the isometry class $G_{2c';\oo,F}$ or $G_{2c';\alpha,F}$,
depending on the value of $\alpha$.

It remains to examine the possibility that $J$ is constant, which by the
evaluation in~\eqref{eq:J-eval-G3ab} restricts to the subclass
$f(\zeta,u) = e^{\Re B} e^{i\gamma} \ln\zeta$, under the two different
branches $\alpha = \oo$ ($\Re B = 0$) and $\alpha \ne \oo$ ($\Re B =
\ln4\alpha - 2\ln u$).

\item \label{itm:G3b-branch} 
Under the $\alpha = \oo$ condition, we are in the class $f(\zeta,u) =
e^{i\gamma} \ln\zeta$. The condition $e^{2i\gamma} = 1$ distinguishes
between the possibilities $\gamma \in \{0,\pi\}$ (when satisfied) and
$\gamma \in (0,\pi)$ (when not). The latter case by the arguments
in~IC.\ref{itm:G3b} isolates the $G_{3b^\circ;\gamma}$ isometry class.
In the remaining possibility, $e^{i\gamma} = \pm 1$, we can evaluate
\begin{equation}
	-2\bD = \pm \bell \bell .
\end{equation}
We have found no further tensor or scalar invariants that can directly
distinguish between the two signs, but they do not give isometric
geometries. However, it seems that an invariant way of distinguishing
between the two signs. Namely, in the $+$ case $-2\bD$ is the
square a real null vector $\bell$, while in $-$ case $-2\bD$ is
the negative square of a real null vector $\bell$. Recall that the
aligned null vector $\bell$ is not actually known at this stage of our
classification, so its existence is a separate condition. By the
arguments from~IC.\ref{itm:G3b0} These conditions respectively isolate
the $G_{3b';\pm}$ isometry classes.

\item \label{itm:G3a-branch} 
Under the $\alpha \ne \oo$ condition, we are in the class $f(\zeta,u) =
\frac{4\alpha}{u^2} e^{i\gamma} \ln\zeta$. The condition $e^{2i\gamma} =
1$ distinguishes between the possibilities $\gamma \in \{0,\pi\}$ (when
satisfied) and $\gamma \in (0,\pi)$ (when not). The latter case by the
arguments in~IC.\ref{itm:G3a} isolates the $G_{3a^\circ;\alpha,\gamma}$
isometry class. In the remaining possibility, $e^{i\gamma} = \pm 1$, we
recall the evaluation~\eqref{eq:Id-eval-G3a0},
\begin{equation}
	\frac{1}{\Id^{(0)}} = \pm \alpha ,
\end{equation}
whose value captures both the $\alpha>0$ and $\pm$ sign invariants. By
the arguments from~IC.\ref{itm:G3a0} this condition isolates the
$G_{3a';\pm\alpha}$ isometry class.

\item 
There are only two possibilities to examine, either $\fL = 0$ or $\fL
\ne 0$.

\item \label{itm:G2b00-branch} 
We are in the formfactor class $f(\zeta,u) = A(u) \ln\zeta + A_1(u)
\zeta + A_0(u)$. Taking into account the evaluation
in~\eqref{eq:fL-def}, the condition $\fL = 0$ is equivalent to
\begin{equation}
	-\frac{1}{2} \frac{\dot{A}}{A} = 0
\end{equation}
or $A(u) = a^2 e^{i\gamma}$ constant. As we already know, the
coefficient $A_0(u)$ is pure gauge, while the coefficient of $\zeta$ must
take the form $A_1(u) = A_{1,0} e^{i\lambda u}$ to correspond to one of
the highly symmetric isometry classes. Consider the
covariant condition
\begin{equation} \label{eq:G2b-log-generic}
	\left(\nabla N - \frac{N}{2} \overline{\bK}\right)^{\otimes 2}
	= 2\lambda e^{-i\gamma \sgn\lambda} N^2 \bD^{\dagger} .
\end{equation}
with external parameters $\lambda \in \mathbb{R}$ and $\gamma \in
[0,\pi]$; when the the condition is satisfied with $\lambda \ne 0$, it
uniquely fixes the values of $\lambda$ and $\gamma$, in the given
ranges, for which it holds. Within our formfactor class, both sides of
the equality are proportional to $\bell\bell$, while equating the
coefficients gives the condition
\begin{equation}
	\dot{A}_1^2 + \lambda^2 A_1^2 = 0 .
\end{equation}
Either it does not hold and we end up in the generic $G_1$ class, or it
does and it is solved by $A_1(u) = A_{1,0} e^{\pm i\lambda u +
i\gamma}$, where $A_{1,0} \in \mathbb{C}$ is an integration constant. We
can use the $f\mapsto \bar{f}$, $u\mapsto -u$ and $\zeta \mapsto
e^{i\beta} \zeta$ gauge freedoms to restrict to $\gamma \in [0,\pi]$,
$A_{1,0} = c \ge 0$ and the $e^{+i\lambda u}$ possibility.

When $\lambda = 0$, the formula in~\eqref{eq:G2b-log-generic} no longer
depends on $\gamma$ and hence cannot determine its value. In that case,
as long as $A(u) = a^2 e^{i\gamma}$, the following covariant relation
must hold
\begin{equation}
	\bD = e^{2i\gamma} \overline{\bD} ,
\end{equation}
which uniquely determines the value of $\gamma \in [0,\pi]$ as long as
$e^{2i\gamma} \ne 1$. In the latter case, using the evaluation
in~\eqref{eq:N-eval-G2c}, we can differentiate
between the values
\begin{equation}
	-2\bD = +(a\bell)^{\otimes 2} \quad \text{and} \quad
		{-(a\bell)^{\otimes 2}}
\end{equation}
by whether $-2\bD$ is the positive ($\gamma=0$) or negative
($\gamma=\pi$) square of a real null vector ($a\bell$), respectively.
Finally, again relying on the evaluations in~\eqref{eq:N-eval-G2c}, we
can use the following formula to extract the invariant parameter
\begin{equation}
	c = |A_{1,0}| = \frac{|N|}{|\bK|} .
\end{equation}
In this branch of the classification, we are guaranteed to have
$c > 0$, because the $c = 0$ case was already handled by
the branch containing the $G_{3b}$ classes. Thus, we have isolated the
$G_{2b'';\gamma,\lambda,c}$ isometry class.

\item \label{itm:G2a0-branch} 
We are in the formfactor class $f(\zeta,u) = A(u) \ln\zeta + A_1(u)\zeta
+ A_0(u)$. The condition
\begin{equation}
	\frac{1}{\Id^{(i\kappa)}} = \alpha e^{i\gamma}
	\quad \iff \quad
	\frac{A^3}{\dot{A}^2} = \alpha e^{i\gamma} ,
\end{equation}
with external parameters $\alpha > 0$, $\gamma\in [0,\pi]$, integrates
to $A(u) = \frac{4\alpha}{(u-u_0)^2} e^{i\gamma}$, where we can set
$u_0=0$ by gauge freedom. Note that the any value of $\kappa$ produces
the same result and $\Id^{(i\kappa)}$ is well-defined
in~\eqref{eq:Idk-def} because we are in the $\Ic=\frac{1}{2}$ branch.

Then, taking into account the evaluation in~\eqref{eq:M-eval-G2a}, the
condition
\begin{equation}
	\bM^{(i\kappa)} = 0 \quad \iff \quad
	-\frac{4}{u\bar{\zeta}} (u\dot{A}_1 - (i\kappa-2) A_1) \bell\bell\bell
	= 0 ,
\end{equation}
with external parameter $\kappa \in \mathbb{R}$ that now enters
non-trivially, integrates to $A_1(u) = A_{1,0} u^{i\kappa-2}$. We can
set $A_{1,0} = c > 0$ by the $\zeta \mapsto e^{i\beta}\zeta$ gauge freedom,
while the $f\mapsto \bar{f}$ freedom has already been fixed by the
requirement that $\gamma \in [0,\pi]$. If either of the above two
conditions fails, we end up in the generic $G_1$ class. Otherwise,
relying on the evaluation in~\eqref{eq:N-eval-G2c} and using the
following formula to extract the invariant parameter
\begin{equation}
	c = |A_{1,0}| = \frac{|N|}{|\bK|} ,
\end{equation}
we have isolated the $G_{2a';\alpha,\gamma,\kappa,c}$ isometry
class.

\item 
By the evaluation~\eqref{eq:fL-def}, the condition $\fL = 0$ is
equivalent to
\begin{equation}
	\left(\frac{(\ln f_{,\zeta\zeta})_{,u}}{(\ln
	f_{,\zeta\zeta})_{,\zeta}}\right)_{\zeta} = 0 .
\end{equation}
This equation has already been integrated in~IC.\ref{itm:G2b}
and~IC.\ref{itm:G2b0} around~\eqref{eq:G2b-entry}
and~\eqref{eq:G2b0-entry}. Since the $\Ic = -(1/(\ln
f_{,\zeta\zeta})_{,\zeta}) = \frac{1}{2}$ possibility has already been
handled in another branch, we can assume that $\Ic \ne \frac{1}{2}$ and
as a result $f(\zeta,u) = \ff(\zeta) + A_1(u) \zeta + A_0(u)$. The
additional condition $\bM^{(\oo)} = 0$, relying on
evaluation~\eqref{eq:M-eval-G2b}, is equivalent to
\begin{equation}
	\ff'''(z) = 0  \quad \text{or} \quad
	\dot{A}_1 = 0 .
\end{equation}
The first possibility has already been handled in a previous branch
where $\overline{\bK}\cdot\bK = 0$, so we can assume $\ff'''(z) \ne 0$,
which then restricts $A_1(u) = A_{1,0}$ constant. Recalling that
$A_0(u)$ is pure gauge, we end up in the formfactor class
\begin{equation}
	f(\zeta,u) = \ff(\zeta) .
\end{equation}
We do not need to separately consider the form $\ff(\zeta) +
A_{1,0}\zeta$ because we can simply redefine $\ff(z)$ to absorb the term
proportional to $A_{1,0}$.

If $\fL = 0$, but $\bM^{(\oo)} \ne 0$, while $\Ic \ne \frac{1}{2}$, then
the underlying geometry is not one of the highly symmetric pp-wave
classes, which lands us in the generic $G_1$ class.

It remains to examine the possibility that $\fL \ne 0$.

\item \label{itm:exp-pow-branch} 
Within the formfactor class $f(\zeta,u) = \ff(\zeta)$, relying on the
evaluation~\eqref{eq:Ide-eval-G2b} with $\lambda=0$, consider the
condition
\begin{equation} \label{eq:Ieoo-G2b0}
	\Ie^{(\oo)} = \frac{\ff'''(\zeta) \ff'(\zeta)}{\ff''(\zeta)^2}
	= \frac{k-2}{k-1} ,
\end{equation}
where the expression on the right-hand side is just a convenient way for
us to parametrize a complex constant. To exhaust all the possible values
of the constant, we must include $k=\oo$ ($\Ie^{(\oo)} = 1$) and
obviously exclude $k=1$.

When $k=\oo$, the corresponding equation has been integrated
in~IC.\ref{itm:G3c}, which isolates the $G_{3c;\lambda}$ isometry
class, with $\lambda = \frac{1}{2} |\bK| > 0$, using the
evaluation~\eqref{eq:K2-def}.

When $k\ne\oo$, the integration in~IC.\ref{itm:G2b-zk} yields two
special values $k=0$ and $k=2$ (with $k=1$ already excluded above), and
the special subfamily $k=2i\kappa$, $\kappa \in \mathbb{R} \setminus
\{0\}$. The possibility $k=0$ is already handled in the $\Ic =
\frac{1}{2}$ branch, while the possibility $k=2$ is already handled in
the $\overline{\bK} \cdot \bK = 0$ branch. As discussed
in~IC.\ref{itm:G3d}, the special subfamily $k=2i\kappa$ isolates
the $G_{3d;\gamma,\kappa}$ isometry class, with $\gamma \in (0,\pi)$
uniquely fixed by
\begin{equation} \label{eq:G3d-gamma}
	e^{2i\gamma}
	= \frac{(2i\kappa+1) (i\kappa-1)^2}{(2i\kappa-1) (i\kappa+1)^2}
		\left|\frac{\bK}{2(i\kappa-1)}\right|^{4i\kappa} J ,
\end{equation}
relying on the evaluation in~\eqref{eq:J-eval-G3d}. This invariant isolates $\gamma$ in every case except when $\gamma = 0$ or $\pi$. In that case both values of $\gamma$ are fixed by complex conjugation and give $e^{2i\gamma} = 1$. We have found that it is not possible to isolate $e^{i\gamma}$ directly using tensor invariants. However, using the Cartan-Karlhede algorithm, it is possible to isolate $e^{i\gamma}=\pm 1$ and the corresponding isometry classes are definitely distinct. So for this particular formfactor and choice of $\gamma$, the IDEAL characterization is unable to distinguish between them.

The remaining values of $k\ne i\mathbb{R} \cup \{1,2\}$, according to
the discussion in~IC.\ref{itm:G2b-zk} isolates the $G_{2b''';k}$
isometry class.

It remains to explore the possibility that $\Ie^{(\oo)} \ne$ const.

\item \label{itm:G2b0-branch} 
Within the $f(\zeta,u) = \ff(\zeta)$ formfactor family, $\Ie^{(\oo)}$
depends holomorphically on $\zeta$, as we see from~\eqref{eq:Ieoo-G2b0}.
Further, in this branch, we can assume that $\Ie^{(\oo)} \ne $ const.
Then, from evaluation~\eqref{eq:K2-def}, $|\bK|^2 = \overline{\cK} \cK$,
where $\cK = \ff'''(\zeta)/\ff''(\zeta)$ obviously also depends
holomorphically on $\zeta$. By the holomorphic implicit function
theorem, there must then exist some holomorphic function $F(y)$ such
that $\ln\cK = F(\Ie^{(\oo)})$, which coincides with the signature
relation~\eqref{eq:G2b0-signature} in~IC.\ref{itm:G2b0}.
Expressed in terms of the $|\bK|$
invariant the relation becomes
\begin{equation}
	\ln(|\bK|^2) = 2 \Re F\left(\Ie^{(\oo)}\right) ,
\end{equation}
where we recognize the harmonic function $\Re F(y)$ as one of the
invariant isometry class parameters discussed in~IC.\ref{itm:G2b0}.
The remaining invariant parameter discussed in~IC.\ref{itm:G2b0} is
isolated by the following relation, where the function $\cG_{\Re F}(w)$
from~\eqref{eq:GF-def} is determined by $\Re F$ (through holomorphic
extension and solving an ODE),
\begin{equation}
	\cG_{\Re F}\left(\Ie^{(\oo)}\right) = \cos\gamma \in [-1,1]
\end{equation}
with the right-hand side constant. Together, these parameters isolate
the isometry class $G_{2b';\gamma,\Re F}$, with $\gamma \in [0,\pi]$
uniquely fixed by the value $\cos\gamma \in [-1,1]$.

\item 
In this branch, we only have the inequalities $\overline{\bK}\cdot \bK
\ne 0$, $\fL \ne 0$ and $\Ic - \frac{1}{2} \ne 0$s. To define further
invariants that we could use for further analysis, as indicated in
Section~\ref{sec:notation}, we also need to eliminate the condition
\begin{equation} \label{eq:Ic-denom-cond}
	k (\Ic - \tfrac{1}{2}) \ne \Ic
	\quad \implies \quad
	\left(\frac{1}{(\ln f_{,\zeta\zeta})_{,\zeta}}\right)_{,\zeta}
		= -\frac{k}{2(k-1)}
\end{equation}
for some complex constant $k$, where we have used the evaluation
in~\eqref{eq:Ic-def}. The general 4-parameter solution for different
values of $k$ (interpreting the special value $k=\oo$ as
$\Ic-\frac{1}{2}=0$, and $k=0$ as $\Ic=0$) is given by
\begin{equation}
	f(\zeta,u) = \begin{cases}
		A \ln(\zeta-h) + A_1 \zeta + A_0
		& (k=\oo) , \\
		e^{A(\zeta-h)} + A_1 \zeta + A_0
		& (k=0) , \\
		A (\zeta-h)^{2/k} + A_1 \zeta + A_0
		& (k\ne 0,\oo) ,
	\end{cases}
\end{equation}
where the integration parameters $A$, $A_1$, $A_0$ and $h$ are of course
$u$-dependent. None of these forms belong to our catalog of highly
symmetric pp-wave geometries in~\eqref{eq:f-cases}, unless they have
already been handled in other branches of our decision diagram in
Figure~\ref{fig:pp-wave-flowchart}, namely node~\ref{itm:log-branch} for
the logaritmic form or node~\ref{itm:exp-pow-branch} for the exponential
and power forms.

Thus, if the equality in~\eqref{eq:Ic-denom-cond} is actually satisfied,
the pp-wave geometry must belong belong to the generic $G_1$ class, and
it remains to explore the possibility where the equality does not hold.

\item 
Within this branch, the $\bL^{(k)}$ invariant is well-defined. Recall
the evaluation $\bL^{(k)} = -2\Xi^{(k)} \bell$ in~\eqref{eq:Lk-def},
where $\Xi^{(k)}$ is evaluated in~\eqref{eq:Xik-def}. Recall also that in
our normalization, the null vector $\bell$ is covariantly constant,
which makes it easy to evaluate the covariant condition
\begin{equation} \label{eq:G2a-Lk-cond}
	2\nabla \bL^{(k)} = \bL^{(k)} \bL^{(k)}
	\quad \iff \quad
	(\Xi^{(k)})_{,\zeta} = 0, \quad
	(\Xi^{(k)})_{,u} + (\Xi^{(k)})^2 = 0 .
\end{equation}
These differential equations are easily integrated to the general
solution
\begin{equation}
	\Xi^{(k)} = \frac{1}{u-u_0} ,
\end{equation}
where $u_0$ is a real integration constant.

Let us suppose that condition~\eqref{eq:G2a-Lk-cond} is satisfied with
$k=i\kappa$, $\kappa \in \mathbb{R}$. For completeness, we must also
consider the possibility that $\kappa = \oo$, which means that the limit
$\bL^{(\oo)} = \lim_{k\to \oo} \bL^{(k)}$ exists and the appropriately
scaled version of the invariant condition~\eqref{eq:G2a-Lk-cond} is
satisfied. We will examine the $\kappa = \oo$ possibility separately.

Provided $\kappa \ne \oo$, using the evaluation in~\eqref{eq:Xik-def}
and clearing the denominators in $\Xi^{(k)} = 1/(u-u_0)$
straightforwardly results in equation~\eqref{eq:G2a-entry}, which was
integrated in~IC.\ref{itm:G2a}. The formfactor $f(\zeta,u)$ then takes
the form
\begin{equation}
	f(\zeta,u) = \frac{1}{u^2} \ff(\zeta u^{i\kappa}) + A_1(u) \zeta ,
\end{equation}
where we have already applied allowed simplification by residual gauge
freedom. A similar integration is possible for other complex values of
$k\in \mathbb{C} \setminus i\mathbb{R}$, but results in a form that can
only belong to the generic $G_1$ class. If the
condition~\eqref{eq:G2a-Lk-cond} is not satisfied at all (including with
$\kappa=\oo$), the underlying pp-wave geometry cannot belong to any of
the highly symmetric classes, so it must again belong to the generic
$G_1$ class.

Since the invariant $\bL^{(i\kappa)}$ is well-defined, so is
$\bM^{(i\kappa)}$. According to its evaluation in~\eqref{eq:M-eval-G2a},
the following invariant condition specializes as
\begin{equation} \label{eq:G2a-Mk-cond}
	\bM^{(i\kappa)} = 0
		\quad \iff \quad
	u\dot{A}_1 - (i\kappa-2) A_1 = 0 ,
\end{equation}
which results in the form factor $f(\zeta,u) = u^{-2} (\ff(\zeta
u^{i\kappa}) + A_{1,0} \zeta u^{i\kappa})$, meaning that the
$\zeta$-linear term can just be absorbed by a redefinition of $\ff(z)$,
which lands us in the $G_{2a}$ subclass. Otherwise, if
condition~\eqref{eq:G2a-Mk-cond} is not satisfied at all, the underlying
pp-wave geometry cannot belong to any of the highly symmetric classes,
so it must again belong to the generic $G_1$ class.

Within the $G_{2a}$ subclass, as evaluated in~\eqref{eq:Ide-eval-G2a},
\begin{equation}
	\Id^{(i\kappa)} = -\frac{\ff'''(z)^2}{\ff''(z)^3}
		\quad \text{and} \quad
	|\bK|^2 = \frac{\overline{\ff'''(z)}}{\overline{\ff''(z)}}
			\frac{\ff'''(z)}{\ff''(z)} ,
\end{equation}
where $z = \zeta u^{i\kappa}$. Clearly, $\Id^{(i\kappa)}$ depends
holomorically on $z$, while $|\bK|^2 = \overline{\cK} \cK$, where $\cK =
\ff'''(z)/\ff''(z)$ holomorphically depends on $z$.

Consider the possibility that $\Id^{(i\kappa)} = \text{const}$. The
simplest case $\Id^{(i\kappa)} = 0$ is equivalent to $\ff'''(z) = 0$ is
integrated by the general solution $\ff(z) = Az^2 + A_1 z + A_0$. All
corresponding pp-wave classes fall into the $G_5$ subclass and have
already been classified in an earlier branch (namely
node~\ref{itm:z2-branch}) of our decision diagram in
Figure~\ref{fig:pp-wave-flowchart}. Otherwise, the condition
$\Id^{(i\kappa)} \ne 0$ appears as~\eqref{itm:G2a0-entry}, where it has
been integrated in IC.\ref{itm:G2a0}--\ref{itm:G3a0} and shown
correspond to either $G_{2a'}$ or $G_{3a}$ subclasses, which have also
been classified in an earlier branch (namely nodes~\ref{itm:G3a-branch}
and~\ref{itm:G2a0-branch}) of our decision diagram.

Thus, in this branch we can assume that $\Id^{(i\kappa)} \ne
\text{const}$ and hence by the holomorphic implicit function theorem,
there must then exist some holomorphic function $F(y)$ such that $\ln\cK
= F(-\Id^{(i\kappa)})$, which coincides with the signature
relation~\eqref{eq:G2a-signature}. Expressed in terms of the $|\bK|$
invariant, the relation becomes
\begin{equation}
	\ln(|\bK|^2) = 2 \Re F\left(-\Id^{(i\kappa)}\right) ,
\end{equation}
where we recognize the harmonic function $\Re F(y)$ as one of the
invariant isometry class parameters discussed in~IC.\ref{itm:G2a}. The
remaining invariant parameter discussed in~IC.\ref{itm:G2a} is isolated
by the following relation, where the function $\cA_{\kappa, \Re F}(w)$
from~\eqref{eq:AkF-def} is determined by $\Re F$ (through holomorphic
extension and solving an ODE),
\begin{equation}
	\cA_{\kappa,\Re F}(\Id^{(i\kappa)}, \Ie^{(i\kappa)})
	= c > 0
\end{equation}
with the right-hand side constant. Together, these parameters isolate
the isometry class $G_{2a^\circ;\kappa,c,\Re F}$.

\item 
Within this branch, the $\bL^{(\oo)}$ invariant is well-defined. Recall
the evaluation $\bL^{(\oo)} = -2\Xi^{(\oo)} \bell$
in~\eqref{eq:Loo-def}, where $\Xi^{(\oo)}$ is evaluated
in~\eqref{eq:Xioo-def}. Recall also that in our normalization, the null
vector $\bell$ is covariantly constant, which makes it easy to evaluate the covariant
condition (which is the $k\to \oo$ scaled limit
of~\eqref{eq:G2a-Lk-cond})
\begin{equation} \label{eq:G2b-Loo-cond}
	2\nabla \bL^{(\oo)} = 0
	\quad \iff \quad
	(\Xi^{(\oo)})_{,\zeta} = 0, \quad
	(\Xi^{(\oo)})_{,u} = 0 .
\end{equation}
These differential equations are easily integrated to the general
solution
\begin{equation}
	\Xi^{(\oo)} = i\lambda ,
\end{equation}
where $\lambda$ is a complex integration constant. Shortly, $\lambda$
will be restricted to real values.

Let us suppose that condition~\eqref{eq:G2b-Loo-cond} is satisfied.

Using the evaluation in~\eqref{eq:Xioo-def}
and clearing the denominators in $\Xi^{(\oo)} = i\lambda$
straightforwardly results in equation~\eqref{eq:G2b-entry}, which was
integrated in~IC.\ref{itm:G2b}. The formfactor $f(\zeta,u)$ then takes
the form
\begin{equation}
	f(\zeta,u) = \ff(\zeta e^{i\lambda u}) + A_1(u) \zeta ,
\end{equation}
where we have already applied allowed simplification by residual gauge
freedom. At this point, we see that only real values of $\lambda$ would
be allowed within the $G_{2b}$ subclass. Hence complex values of
$\lambda \in \mathbb{C} \setminus \mathbb{R}$ result in a form that can
only belong to the generic $G_1$ class. Since $\bL^{(\oo)} = -2i\lambda
\bell$, we can impose $\lambda \in \mathbb{R}$ by the invariant
condition
\begin{equation}
	\Re \bL^{(\oo)} = 0 .
\end{equation}
If the condition~\eqref{eq:G2b-Loo-cond} is not satisfied at all, the
underlying pp-wave geometry cannot belong to any of the highly symmetric
classes, so it must again belong to the generic $G_1$ class.

Since the invariant $\bL^{(\oo)}$ is well-defined, so is
$\bM^{(\oo)}$. According to its evaluation in~\eqref{eq:M-eval-G2b},
the following invariant condition specializes as
\begin{equation} \label{eq:G2b-Moo-cond}
	\bM^{(\oo)} = 0
		\quad \iff \quad
	\dot{A}_1 - i\lambda A_1 = 0 ,
\end{equation}
which results in the form factor $f(\zeta,u) = \ff(\zeta e^{i\lambda u})
+ A_{1,0} \zeta e^{i\lambda u}$, meaning that the $\zeta$-linear term
can just be absorbed by a redefinition of $\ff(z)$, which lands us in
the $G_{2b}$ subclass. Otherwise, if condition~\eqref{eq:G2b-Moo-cond}
is not satisfied at all, the underlying pp-wave geometry cannot belong
to any of the highly symmetric classes, so it must again belong to the
generic $G_1$ class.

Note that at this point we can assume that $\lambda \ne 0$, since the
$\lambda = 0$ was already handled in an earlier branch (namely
node~\ref{itm:G2b0-branch}) in our decision diagram in
Figure~\ref{fig:pp-wave-flowchart}. Moreover, using the residual gauge
transformation $u\mapsto u/\lambda$, we are free to assume the value
$\lambda = 1$.

Within the $G_{2b}$ subclass, as evaluated in~\eqref{eq:Ide-eval-G2b},
\begin{equation}
	\Id^{(\oo)} = \lambda^2\frac{\ff'''(z)^2}{\ff''(z)^3}
		\quad \text{and} \quad
	|\bK|^2 = \frac{\overline{\ff'''(z)}}{\overline{\ff''(z)}}
			\frac{\ff'''(z)}{\ff''(z)} ,
\end{equation}
where $z = \zeta e^{i\lambda u}$. Clearly, $\Id^{(\oo)}$ depends
holomorically on $z$, while $|\bK|^2 = \overline{\cK} \cK$, where $\cK =
\ff'''(z)/\ff''(z)$ holomorphically depends on $z$.

Consider the possibility that $\Id^{(\oo)} = \text{const}$ and recall
that $\lambda \ne 0$. The simplest case $\Id^{(\oo)} = 0$ is equivalent
to $\ff'''(z) = 0$ is integrated by the general solution $\ff(z) = Az^2
+ A_1 z + A_0$. All corresponding pp-wave classes fall into the $G_5$
subclass and have already been classified in an earlier branch (namely
node~\ref{itm:z2-branch}) of our decision diagram in
Figure~\ref{fig:pp-wave-flowchart}. Otherwise, the condition
$\Id^{(\oo)} \ne 0$ appears as~\eqref{eq:G2b00-entry}, where it has
been integrated in IC.\ref{itm:G2b00}--\ref{itm:G3b0} and shown
correspond to either $G_{2b''}$ or $G_{3b}$ subclasses, which have also
been classified in an earlier branch (namely nodes~\ref{itm:G3b-branch}
and~\ref{itm:G2b00-branch}) of our decision diagram.

Thus, in this branch we can assume that $\Id^{(\oo)} \ne \text{const}$
and hence by the holomorphic implicit function theorem, there must then
exist some holomorphic function $F(y)$ such that $\ln\cK =
F(\Id^{(\oo)})$, which coincides with the signature
relation~\eqref{eq:G2b-signature} (recall the simplification $\lambda =
1$). Expressed in terms of the $|\bK|$ invariant, the relation becomes
\begin{equation}
	\ln(|\bK|^2) = 2 \Re F\left(\Id^{(i\kappa)}\right) ,
\end{equation}
where we recognize the harmonic function $\Re F(y)$ as one of the
invariant isometry class parameters discussed in~IC.\ref{itm:G2b}. The
remaining invariant parameter discussed in~IC.\ref{itm:G2b} is isolated
by the following relation, where the function $\cA_{\oo, \Re F}(w)$
from~\eqref{eq:AooF-def} is determined by $\Re F$ (through holomorphic
extension and solving an ODE),
\begin{equation}
	\cA_{\oo,\Re F}(\Id^{(\oo)}, \Ie^{(\oo)})
	= c > 0
\end{equation}
with the right-hand side constant. Together, these parameters isolate
the isometry class $G_{2b^\circ;c,\Re F}$.

\end{enumerate}

Together with the data in Figure~\ref{fig:pp-wave-flowchart} and
Table~\ref{tab:pp-wave-ideal} (which collects the complete IDEAL
characterization identities that can be read off from
Figure~\ref{fig:pp-wave-flowchart} for each isometry class) this
completes the IDEAL characterization of regular highly symmetric pp-wave
geometries.
\end{proof}

\section{Discussion} \label{sec:discuss}

We have investigated the necessity of traditional scalar polynomial curvature invariants in the IDEAL approach to the classification of spacetimes, by considering the class of vacuum pp-wave spacetimes. These spacetimes belong to the larger class of VSI spacetimes for which all polynomial scalar curvature invariants vanish. 

To hold with the primary principles of the IDEAL approach, we relaxed the constraint on the construction of scalar invariants by including conditional invariant tensors and more specifically, conditional invariant scalars as permitted quantities. A conditional invariant tensor or scalar arises when two IDEAL tensors are proportional. With this modification, we are able to characterize all of the highly-symmetric vacuum pp-wave spacetimes in an IDEAL manner and more finely classify these spacetimes into formfactor classes. 

For a given highly symmetric vacuum pp-wave formfactor, there may be parameters or functions that must be fixed to specify a particular metric. This presents the problem of sub-classification within a given formfactor class and whether the IDEAL approach presented here is applicable. We have attempted to determine a minimal set of IDEAL quantities that will pick out the parameters and functions in the metric formfactor for each of the possible symmetry sub-classes. This is possible to do for almost all of the formfactor classes, except in the case of $G_{3d;\kappa, \gamma}$ where the constant $\gamma$ only appears in the form of $e^{2i\gamma}$ in every conditional invariant. Thus, the two cases where $e^{i\gamma} = \pm 1$ cannot be distinguished using an IDEAL approach. In every other subclass, it is possible to uniquely characterize a solution in terms of IDEAL quantities.


\section*{Appendix}

In this appendix, we review some useful coordinate transformations for the modest generalizations of those formfactor classes that require a signature function for the function $A(u) = e^{B(u)}$. In addition to the transformations $$A(u) \left\{ \begin{smallmatrix}
\zeta^k \\ \ln\zeta \end{smallmatrix} \right\} \mapsto \overline{A}(u)
\left\{ \begin{smallmatrix} \zeta^{\bar{k}} \\ \ln\zeta
\end{smallmatrix} \right\},$$
\noindent we have the following expressions:
\begin{equation} \label{eq:A-residual1}
	A(u) \mapsto
	\begin{cases}
	a^2 e^{i k \beta} A\left(a(u+u_0)\right) \zeta^k
			+ a^2 e^{i\beta} A_1\left(a(u+u_0)\right) \zeta
		& (k\ne 0,1,2) \\
		&~ z=A_0 = 0 , \\
	a^2 A\left(a(u+u_0)\right) \ln\zeta
			+ a^2 e^{i\beta} A_1\left(a(u+u_0)\right) \zeta
		& (k=0) \\
		&~ z=A_0 = 0 , \\
	a^2 e^{2i\beta} A\left(a(u+u_0)\right) \zeta^2
		& (k=2) \\
		&~ A_1 = A_0 = 0 ,
	\end{cases}
\end{equation}
where on the right we indicate which integration parameters can be set
to zero by appropriately choosing the gauge parameters $c(u)$, $g(u)$
and $h(u)$ from~\eqref{eq:pp-wave-residual-gauge}, with the dependence
on the remaining gauge parameters $\beta$, $a$ and $u_0$ explicitly
shown. We will need to distinguish two special cases, where $\Re \dot{B}
e^{-\frac{1}{2}\Re B} = \alpha^{-1/2}$ for $\alpha > 0$, in which $\Re
B(u)$ no longer depends on the residual gauge parameter $a$,
\begin{equation} \label{eq:A-residual2}
	A(u) \mapsto
	\begin{cases}
	\frac{\alpha}{(u+u_0)^2} e^{i\Im B\left(a(u+u_0)\right)} \ln\zeta
			+ a^2 A_1\left(a(u+u_0)\right) \zeta
		& (k=0) , \\
	\frac{\alpha}{(u+u_0)^2} e^{i\Im B\left(a(u+u_0)\right)+2i\beta} \zeta^2
		& (k=2) ,
	\end{cases}
\end{equation}
and also the case where $\Re \dot{B} e^{-\frac{1}{2} \Re B} = 0$ and
$\Im \dot{B} e^{-\frac{1}{2} \Re B} = \kappa$ for $\kappa \ge 0$,
\begin{equation} \label{eq:A-residual3}
	A(u) \mapsto
	\begin{cases}
	e^{i\kappa a(u+u_0) + 2\ln a} \ln\zeta
			+ a^2 A_1\left(a(u+u_0)\right) \zeta
		& (k=0) , \\
	e^{i\kappa a(u+u_0) + 2\ln a + 2i\beta} \zeta^2
		& (k=2) .
	\end{cases}
\end{equation}

We will use the residual gauge freedom to set $A_1(u) = A_0(u) = 0$ and
choose the phase of $\lambda$ so that $\lambda > 0$. Hence, the
remaining equation for $\ddot{B}(u)$ integrates to $A(u) = e^{B(u)} =
e^{\beta u + \gamma}$, with complex constants $\beta$, $\gamma$. Under
the remaining residual gauge freedom,
\begin{equation} \label{eq:A-residual4}
	A(u) \mapsto a^2 e^{2\lambda(bu+c)} A\left(a(u+u_0)\right) ,
\end{equation}
with $a$, $u_0$ real and $b$, $c$ complex. The latter two constants
appear because setting $\del_u^2 h(u) = A_1(u)$ determines $h(u)$
uniquely only up to the addition of $bu + c$ with $b$ and $c$ arbitrary
complex constants. Clearly, the integration constants $\beta$ and
$\gamma$ can be absorbed by the residual gauge freedom.
\bibliographystyle{utphys-alpha}
\bibliography{pp-idealbib}

\end{document}